\documentclass{lmcs}
\pdfoutput=1

\usepackage{lastpage}
\lmcsdoi{17}{2}{12}
\lmcsheading{}{\pageref{LastPage}}{}{}%
{Sep.~07,~2019}{Apr.~23,~2021}{}


\usepackage{latexsym}
\usepackage{amsmath}
\usepackage{amssymb}
\usepackage{amsthm}
\usepackage[notext,not1,nomath]{stix}    
\usepackage{bussproofs}
\usepackage[utf8]{inputenc}

\usepackage{tikz-cd}
\tikzset{>=latex}
\tikzcdset{arrow style=tikz}
\tikzcdset{every arrow/.append style = {line width = 0.3pt}}

\theoremstyle{defC}
\newtheorem{exaC}[thm]{Example}

\def\tTrue{{\mathrm{1}}}
\def\tFalse{{\mathrm{0}}}
\def\tVar{{\hbox{\textit{Var}}}}
\def\tLit{{\hbox{\textit{Lit}}}}
\def\tdom{{\hbox{\textit{dom}}}}
\def\rest{{{|}}}    
\def\liff{{\leftrightarrow}}
\def\notpigeonto{{\nrightarrow}}
\def\pigeonto{{\rightarrow}}
\def\pprime{{\prime\prime}}
\def\calP{{\mathcal{P}}}
\def\calQ{{\mathcal{Q}}}
\def\SAT{{\mathrm{SAT}}}
\def\REF{{\mathrm{REF}}}
\def\PAR{{\mathrm{PAR}}}
\def\CC{{\mathrm{CC}}}
\def\TS{{\mathrm{TS}}}
\def\tBC{{\hbox{\rm BC}}}
\def\tRUP{{\hbox{\rm RUP}}}
\def\tRAT{{\hbox{\rm RAT}}}
\def\tSPR{{\hbox{\rm SPR}}}

\def\tPR{{\hbox{\rm PR}}}
\def\tSR{{\hbox{\rm SR}}}
\def\permR{{\pi\hbox{\rm PR}}}
\def\tER{{\hbox{\rm ER}}}
\def\tDBC{{\hbox{\rm DBC}}}
\def\tDRAT{{\hbox{\rm DRAT}}}
\def\tDSPR{{\hbox{\rm DSPR}}}
\def\tDPR{{\hbox{\rm DPR}}}
\def\tDSR{{\hbox{\rm DSR}}}

\def\BCnnv{{\tBC^-}}
\def\RATnnv{{\tRAT^-}}
\def\SPRnnv{{\tSPR^-}}
\def\PRnnv{{\tPR^-}}
\def\SRnnv{{\tSR^-}}
\def\permRnnv{{\permR^-}}

\def\DBCnnv{{\tDBC^-}}
\def\DRATnnv{{\tDRAT^-}}
\def\DSPRnnv{{\tDSPR^-}}
\def\DPRnnv{{\tDPR^-}}
\def\DSRnnv{{\tDSR^-}}

\def\PHP{\mathrm{PHP}}
\def\BPHP{{\mathrm{BPHP}}}

\def\bigor{\bigvee}

\def\dotlor{\,\dot\lor\,}

\newcommand{\IGNORE}[1]{}

\newcommand*\olnot[1]{%
   \vbox{%
     \hrule height 0.5pt
     \kern0.25ex
     \hbox{%
       \ifmmode#1\else\ensuremath{#1}\fi
     }
   }
}

\begin{document}

\title[DRAT and Propagation Redundancy Proofs \texorpdfstring{\\}{} Without New Variables]{DRAT and Propagation Redundancy Proofs \texorpdfstring{\\}{} Without New Variables\rsuper*}
\titlecomment{{\lsuper*}A preliminary
version~\cite{BussThapen:DratAndPr_SAT} of this paper appeared in the Proceedings
of the 2019 Conference on Theory and Applications of Satisfiability Testing (SAT).}

\author[S.~Buss]{Sam Buss\rsuper{a}}
\address{\lsuper{a}Department of Mathematics,
University of California, San Diego,
La Jolla, CA 92093--0112, USA}
\email{sbuss@ucsd.edu}
\thanks{This work was initiated on a visit of the first author to
the Czech Academy of Sciences in July~2018,
supported by ERC advanced grant 339691 (FEALORA).
The first author was also supported
by Simons Foundation grant 578919.
The second author was partially supported by
GA \v{C}R project 19--05497S.
The Institute of Mathematics of the Czech Academy of Sciences is supported
by RVO:67985840.}
\author[N.~Thapen]{Neil Thapen\rsuper{b}}
\address{\lsuper{b}Institute of Mathematics of the
Czech Academy of Sciences,
Prague, Czech Republic}
\email{thapen@math.cas.cz}


\begin{abstract} \noindent
We study the complexity of a range of propositional proof systems
which allow inference rules of the form: from a set of clauses
$\Gamma$ derive the set of clauses $\Gamma \cup \{ C \}$ where, due to some syntactic
condition,
$\Gamma \cup \{ C \}$ is satisfiable if $\Gamma$ is,
but where $\Gamma$ does not necessarily imply~$C$.
These inference rules include
BC, RAT, SPR and PR (respectively short for
\emph{blocked clauses},
\emph{resolution asymmetric tautologies},
\emph{subset propagation redundancy} and
\emph{propagation redundancy}), which arose from work in satisfiability (SAT) solving.
We introduce a new, more general rule SR (\emph{substitution redundancy}).

If the new clause $C$ is allowed to include new variables
then the systems based on these
rules are all  equivalent  to extended resolution.
We focus
on restricted systems that do not allow new variables.
The systems with deletion, where we can delete
a clause from our set at any time, are denoted
$\DBCnnv$, $\DRATnnv$, $\DSPRnnv$, $\DPRnnv$ and $\DSRnnv$.
The systems without deletion are
$\BCnnv$, $\RATnnv$, $\SPRnnv$, $\PRnnv$ and~$\SRnnv$.

With deletion, we show that DRAT${}^-$,
DSPR${}^-$ and DPR${}^-$ are equivalent. By earlier work
of Kiesl, Rebola-Pardo and Heule~\cite{KRPH:erDRAT},
they are also equivalent to DBC${}^-$.
Without deletion,
we show that SPR${}^-$ can simulate
PR${}^-$ provided only short clauses are inferred by SPR inferences.
We also show that many of the well-known
``hard'' principles have small SPR${}^-$ refutations.
These include the pigeonhole principle,
bit pigeonhole principle, parity principle,
Tseitin tautologies and clique-coloring tautologies. SPR${}^-$ can also
handle or-fication and xor-ification, and lifting with an index gadget.
Our final result is an exponential size lower
bound for RAT${}^-$ refutations, giving exponential separations between
RAT${}^-$ and both DRAT${}^-$ and SPR${}^-$.
\end{abstract}

\maketitle

\section{Introduction}\label{sec:Intro}

SAT solvers are routinely used for a range of
large-scale instances of satisfiability. It is widely
realized that when a solver reports that a SAT instance $\Gamma$
is unsatisfiable, it should also produce a \emph{proof} that it is
unsatisfiable. This is of particular importance as SAT solvers become
increasingly complex, combining many techniques, and thus are more
subject to software bugs or even design problems.

The first proof systems proposed for SAT solvers were based
on reverse unit propagation ($\tRUP$,
or $\vdash_1$ in the notation of this paper) inferences~\cite{GoldbergNovikov:verification,VanGelder:RUP}
as this is sufficient to handle both resolution inferences and
the usual CDCL clause learning schemes.
However, $\tRUP$ inferences only support logical implication,
and in particular do not accommodate many ``inprocessing'' rules.
Inprocessing rules support inferences which do not respect
logical implication; instead they only guarantee \emph{equisatisfiability}
where the (un)satisfiability of the set of clauses is preserved~\cite{JHB:inprocessing}. 
Inprocessing inferences have been formalized in terms of sophisticated inference
rules including $\tDRAT$ (\emph{deletion, reverse asymmetric tautology}),
$\tPR$ (\emph{propagation redundancy}), $\tSPR$ (\emph{subset $\tPR$}) in a series of papers including~\cite{JHB:inprocessing,%
HHW:verifying,%
HHW:trimming,%
WHH:DRATtrim}
--- see Section~\ref{sec:inferences} for definitions.
These inference rules can be viewed as introducing clauses that
hold ``without loss of
generality''~\cite{RebolaPardoSuda:SatPreserving}, and thus preserve (un)satisfiability. 
An important feature of these systems
is that they can be used both as proof systems to verify
unsatisfiability, and as inference systems to facilitate
searching for either a satisfying assignment or a proof
of unsatisfiability.\footnote{The deletion rule is very helpful
to improve proof search and can extend the power of the
inferences rules, see Corollary~\ref{coro:RATminusnotSimulateTwo};
however, it must be used carefully to preserve equisatisfiabity.
The present paper only considers \emph{refutation systems}, and thus
the deletion rule can be used without restriction.}

The $\tDRAT$ system is very powerful as it can
simulate extended resolution~\cite{Kullmann:GeneralizationER, KRPH:erDRAT}. This simulation
is straightforward, but depends on $\tDRAT$'s ability
to introduce new variables; we simply show that the usual extension
axioms are RAT\@.
 However, there are a number of results~\cite{HeuleBiere:Variable,HKSB:PRuning,HKB:NoNewVariables,HKB:StrongExtensionFree}
indicating that $\tDRAT$ and $\tPR$ are still powerful when restricted to use few new variables, or
even no new variables. In particular,~\cite{HKSB:PRuning,HKB:NoNewVariables,HKB:StrongExtensionFree} showed that the
pigeonhole principle clauses have short (polynomial size)
refutations in the $\tPR$ proof system. The paper~\cite{HKSB:PRuning} showed that
Satisfaction Driven Clause Learning (SDCL) can discover $\tPR$ proofs
of the pigeonhole principle automatically;
in the application studied by~\cite{HKSB:PRuning},
the SDCL search appears to have exponential runtime, but is much more efficient
than the usual CDCL search.
There are at present no broadly applicable proof search heuristics for how to
usefully introduce
new variables with the extension rule.
It is possible however that there are useful heuristics for searching
for proofs that do not use new variables
in $\tDRAT$ and $\tPR$ and related systems.
For these reasons, $\tDRAT$ and $\tPR$ and related systems (even when new
variables are not allowed) hold the potential
for substantial improvements in the power of SAT solvers.

The present paper extends the theoretical knowledge of
these proof systems viewed as refutation systems.
We pay particular attention to proof systems that
do not allow new variables. The remainder of Section~\ref{sec:Intro}
introduces the proof systems $\tBC$ (\emph{blocked clauses}),
$\tRAT$, $\tSPR$, $\tPR$ and $\tSR$ (\emph{substitution redundancy}).
(Only $\tSR$ is new to this paper.) These
systems have variants which allow deletion, called $\tDBC$,
$\tDRAT$, $\tDSPR$, $\tDPR$ and $\tDSR$.  There are also variants
of all these systems restricted to not allow new variables:
we denote these with a superscript~``$-$''
as $\BCnnv$, $\DBCnnv$, $\RATnnv$, $\DRATnnv$,~etc.

Section~\ref{sec:extended_resolution} studies the relation between
these systems and extended resolution. We show
in particular that any proof system containing $\BCnnv$ and
closed under restrictions simulates extended resolution.
Here a proof system~$\calP$ is said to \emph{simulate} a
proof system~$\calQ$ if any $\mathcal Q$-proof can be converted, in polynomial time, into a
$\calP$-proof of the same result.
Two systems are \emph{equivalent} if they simulate each other;
otherwise they are \emph{separated}.
We also show that the systems discussed above all have equivalent
canonical NP pairs (a coarser notion of equivalence).

Section~\ref{sec:simulations} extends known results
that $\DBCnnv$ simulates $\DRATnnv$~\cite{KRPH:erDRAT}
and that $\tDRAT$, limited to only one extra variable,
simulates $\DPRnnv$~\cite{HeuleBiere:Variable}.
Theorem~\ref{thm:DRATnnvDPRnnv}
proves that
$\DRATnnv$ simulates $\DPRnnv$. As a consequence,
$\DBCnnv$ can also simulate $\DPRnnv$. We then give a partial
simulation of $\PRnnv$ by $\SPRnnv$ --- our size bound is exponential
in the size of the ``discrepancy'' of the $\tPR$ inferences, but in many cases,
the discrepancy will be logarithmic or even smaller.

Section~\ref{sec:Upperbounds} proves new polynomial upper bounds
on the size of $\SPRnnv$ proofs for many of the ``hard'' tautologies
from proof complexity.
(Recall that $\SPRnnv$ allows neither deletion nor the
use of new variables.)
These include the pigeonhole principle,
the bit pigeonhole principle, the parity principle, the clique-coloring
principle, and the Tseitin tautologies. We also show that
obfuscation by or-fication, xor-ification and lifting with a indexing
gadget do not work against
$\SPRnnv$. Prior results gave $\SPRnnv$ proofs for
the pigeonhole principle (PHP)~\cite{HKB:NoNewVariables,HKB:StrongExtensionFree},
and $\PRnnv$ proofs for the
Tseitin tautologies and the 2--1 PHP~\cite{HeuleBiere:Variable}.
These results raise the question of whether $\SPRnnv$ (with no
new variables!) can simulate Frege systems, for instance.  Some
possible principles that might separate $\SPRnnv$ from Frege systems
are the graph PHP principle, 3-XOR tautologies and the even coloring principle;
these are discussed at the end of Section~\ref{sec:Upperbounds}. However,
the even coloring principle does have short $\DSPRnnv$ proofs, and it
is plausible
that the graph PHP principle has short $\SPRnnv$ proofs.

Section~\ref{sec:Lowerbounds} shows that
$\RATnnv$ (with neither new variables nor deletion) cannot simulate
either $\DRATnnv$ (without new variables, but
with deletion) or $\SPRnnv$ (with neither new variables nor deletion).
This follows from a size lower
bound for $\RATnnv$ proofs of the bit pigeonhole principle ($\BPHP$).
We first prove a width lower bound, by showing that any
$\tRAT$ inference in a small-width refutation of $\BPHP$ can be
replaced with a small-width resolution derivation,
and then derive the size bound.
We use that $\BPHP$ behaves well when the sign of a variable is
flipped.

\IGNORE{
Most of the known inclusions for these systems,
including our new results, are summarized in~\eqref{eq:results1}--\eqref{eq:results3}.
Allowing new variables (and with or without deletion), we have
\begin{equation}\label{eq:results1}
\mathrm{Res} <
\tBC \equiv \tRAT \equiv \tSPR \equiv \tPR \equiv \tSR
\equiv \tER.
\end{equation}
With deletion and no new variables (except $\tER$ may use new variables):
\begin{equation}\label{eq:results2}
\mathrm{Res} <
\DBCnnv \equiv \DRATnnv \equiv \DSPRnnv \equiv \DPRnnv \le \DSRnnv
\le \tER.
\end{equation}
With no deletion and no new variables (except $\tER$ may use new variables):
\begin{equation}\label{eq:results3}
\mathrm{Res} <
\BCnnv \le \RATnnv < \SPRnnv \le^* \PRnnv \le \SRnnv
\le \tER.
\end{equation}
In these equations, equivalence ($\equiv$) indicates the systems
simulate each other. Inequality ($\le$) indicates only one direction
is known for the simulation. Strict inequality ($<$) means that it is known
there is no simulation in the other direction.
The symbol $\le^*$ in~(\ref{eq:results3})
means $\PRnnv$ simulates $\SPRnnv$, and
there is  a simulation in the other direction under the additional
assumption that the discrepancies (see
Definition~\ref{def:discrepancy}) of $\tPR$ inferences are logarithmically bounded.
}

The known relationships between these systems,
including our results, are summarized in Figure~\ref{fig:diagram}.
Recall that e.g. $\tBC$ is the full system,
$\DBCnnv$ is the system with deletion but no new variables,
and $\BCnnv$ is the system with neither deletion nor new variables.
An arrow shows that the upper system simulates the lower one.
Equivalence $\equiv$ indicates that the systems simulate each other.
The arrow from $\PRnnv$ and $\SPRnnv$ is marked $*$
to indicate that
there is  a simulation in the other direction under the additional
assumption that the discrepancies (see
Definition~\ref{def:discrepancy}) of $\tPR$ inferences are logarithmically bounded.

\begin{figure}
\begin{tikzcd}[row sep = 0.6cm]
& \makebox[2em]{$\tER \equiv \tSR \equiv \tPR \equiv \tSPR \equiv \tRAT
\equiv \tBC$} \arrow[d] & \\
& \DSRnnv \arrow[ld] \arrow[rd] & \\
\makebox[4em]{$\DPRnnv \equiv \DSPRnnv \equiv \DRATnnv \equiv \DBCnnv$}\arrow[rd]
	&  & \SRnnv \arrow[ld]\\
& \PRnnv \arrow[d, "~*" near start] \\
& \SPRnnv \arrow[d, "~\not\equiv" near start] \\
& \RATnnv \arrow[d] \\
& \BCnnv \arrow[d, "~\not\equiv" near start] \\
& \mathrm{Res}
\end{tikzcd}
\caption{Relationships between proof systems.}\label{fig:diagram}
\end{figure}
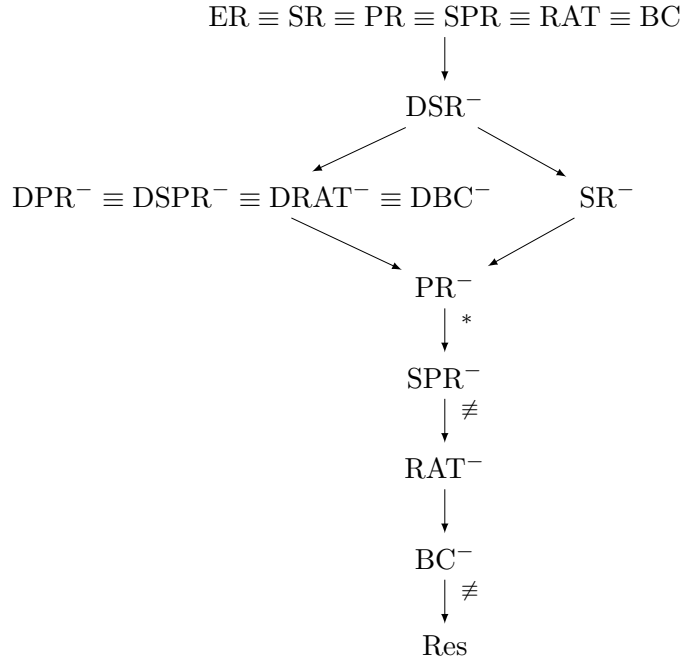

We summarize the rules underlying these systems in Table~\ref{table:rules}. The details
and the necessary definitions
are in Section~\ref{sec:inferences} below --- in particular
see Theorem~\ref{thm:RATandLPR} for this definition of RAT\@.

\begin{table}
\vspace{6mm}
\begin{center}
\begin{tabular}{ c  l  l }
  \tBC
  	& (a restriction of \tRAT)
	 & \emph{blocked clause} \\[0.5mm]
  \tRAT &  	
   $\tau$ is $\alpha$ with one variable flipped
  & \emph{reverse asymmetric tautology} \\[0.5mm]
%
  \tSPR & $\tau$ is a partial assignment,
  	$\tdom(\tau) \! = \! \tdom(\alpha)$
  	& \emph{subset propagation redundant} \\[0.5mm]
  \tPR &$\tau$ is a partial assignment
  	& \emph{propagation redundant} \\[0.5mm]
  \tSR & no extra conditions
  	&  \emph{substitution redundant}\\
\end{tabular}
\end{center}
\caption{Summary of rules of inference.}\label{table:rules}
\end{table}

As presented here the 
rules (except for BC) have the form:
derive $C$ from $\Gamma$, if there is a substitution $\tau$
satisfying $\Gamma_{\rest\alpha} \vdash_1 \Gamma_{\rest\tau}$
plus the conditions shown,
where $\alpha$ is $\olnot{C}$. The implication $\vdash_1$ is defined below in
terms of reverse unit propagation (RUP).

We remark that the question of whether new variables help
reasoning with blocked clause inferences was already studied by Kullmann in the
context of the system
Generalized Extended Resolution (GER)~\cite{Kullmann:GeneralizationER}.
As far as we know, GER does not correspond exactly to any of the systems we consider.
\cite{Kullmann:GeneralizationER}~showed that
allowing new variables does not reduce GER proof length
when the blocked clause rule is restricted to
introducing clauses of length at most two.

\subsection{Preliminaries}\label{sec:prelim}

We use the usual conventions for clauses, variables,
literals, truth assignments, satisfaction, etc.
$\tVar$ and $\tLit$ denote the sets of all
variables and all literals.
A set of literals is called \emph{tautological} if it contains a pair of
complementary literals $p$ and~$\olnot p$.
A \emph{clause} is a non-tautological\footnote{
Disallowing tautological clauses makes the rest
of the definitions more natural. In particular, we
can identify clauses with the negations of partial assignments.}
 set of literals;
we use $C, D, \ldots$ to denote clauses.
The empty clause is denoted~$\perp$, and is always false.
$\tFalse$ and $\tTrue$ denote respectively
\emph{False} and \emph{True}; and
$\olnot \tFalse$ and~$\olnot \tTrue$ are respectively $1$ and~$0$.
We use both $C\cup D$ or $C\lor D$ to denote unions of
clauses, but usually write $C\lor D$ when the union is a clause.
The notation $C = D \dotlor E$
indicates that $C= D\lor E$ is a clause and $D$ and $E$ have no variables in common.
If $\Gamma$ is a set of clauses, $C \lor \Gamma$
is the set $\{C \lor D: \text{$D \in \Gamma$ and $C\lor D$ is a clause}\}$.

A \emph{partial assignment}~$\tau$ is  a mapping
with domain a set of variables and range
contained in $\{\tFalse,\tTrue\}$.
It acts on literals by letting
$\tau(\olnot p) = \olnot{\tau(p)}$.
It is called a \emph{total assignment} if it sets all variables.
We sometimes identify a partial assignment $\tau$
with the set of unit clauses asserting that $\tau$ holds.
For $C$ a clause, $\olnot{C}$
denotes the partial assignment whose domain is the variables of $C$
and which asserts that $C$ is false.
For example, if $C = x \lor \olnot y \lor z$
then, depending on context, $\olnot C$ will denote either the
set containing the three unit clauses $\olnot x$ and~$y$ and~$\olnot z$,
or the partial assignment $\alpha$ with domain $\tdom(\alpha) = \{x,y,z\}$ such that
$\alpha(x) = \tFalse$, $\alpha(y) = \tTrue$
and $\alpha(z) = \tFalse$.

A \emph{substitution}  generalizes the notion of
a partial assignment by allowing variables to be mapped
also to literals.
Formally,
a substitution~$\sigma$ is a map from $\tVar \cup \{\tFalse, \tTrue \}$ to
$\tLit \cup \{\tFalse,\tTrue\}$ which is
the identity on $\{\tFalse, \tTrue \}$.
Note that a
substitution may cause different literals to become
identified.\footnote{\cite{Szeider:Homomorphisms} defined
a notion of ``homomorphisms'' that is similar to substitutions.
Substitutions, however, allow variables to be mapped also to constants.
Our SR inference, defined below, uses $\vdash_1$;
this was not used with homomorphisms in~\cite{Szeider:Homomorphisms}.}
A partial assignment~$\tau$ can be viewed as a
substitution, by defining $\tau(x)=x$ for all
variables $x$ outside the domain of $\tau$.
The \emph{domain} of a substitution $\sigma$ is the
set of variables~$x$ for which $\sigma(x) \neq x$.

Suppose $C$ is a clause and $\sigma$~is a substitution (or a
partial assigment viewed as a substitution).
Let $\sigma(C) = \{ \sigma(p) : p \in C \}$.
We say $\sigma$ \emph{satisfies} $C$, written $\sigma \vDash C$, if
$1\in \sigma(C)$ or $\sigma(C)$ is tautological.
When $\sigma \nvDash C$, the
\emph{restriction} $C_{\rest \sigma}$ is defined
by letting $C_{\rest \sigma}$ equal $\sigma(C) \setminus \{\tFalse\}$.
Thus $C_{\rest \sigma}$ is a clause
expressing the meaning of~$C$ under~$\sigma$.
For $\Gamma$ a set of
clauses, the restriction of $\Gamma$ under~$\sigma$ is
\[
\Gamma_{\rest \sigma} ~=~
   \{\, C_{\rest\sigma} : \text{ $C\in \Gamma$ and $\sigma \nvDash C$ } \}.
\]

The composition of two substitutions is
denoted $\tau\circ\pi$, meaning that
$(\tau\circ\pi)(x) = \tau(\pi(x))$,
and  in particular $(\tau\circ\pi)(x) = \pi(x)$ if $\pi(x)\in\{\tFalse,\tTrue\}$.
For partial assignments $\tau$ and~$\pi$, this
means that $\tdom(\tau\circ\pi) = \tdom(\tau)\cup\tdom(\pi)$ and
\[
(\tau\circ\pi)(x) ~=~ \left\{ \begin{array}{ll}
\pi(x) ~~~~ & \hbox{if $x\in\tdom(\pi)$} \\
\tau(x) & \hbox{if $x\in\tdom(\tau)\setminus\tdom(\pi)$.}
\end{array}\right.
\]

\begin{lem}\label{lem:substitution_composition}
For a set of clauses $\Gamma$ and substitutions $\tau$ and $\pi$,
$\Gamma_{\rest \tau \circ \pi} =
{(\Gamma_{\rest \pi})}_{\rest\tau}$.
In particular, $\tau \vDash \Gamma_{\rest\pi}$ if and only if
$\tau \circ \pi \vDash \Gamma$.
\end{lem}

\begin{proof}
Notice $\tau \circ \pi \vDash C$ if and only if
$\pi \vDash C$ or $( \pi \not \vDash C \wedge \tau \vDash C_{\rest \pi})$.
Thus
\begin{align*}
{(\Gamma_{\rest \pi})}_{\rest\tau}
&= \big\{ {(C_{\rest \pi})}_{\rest\tau} : C \in \Gamma, \
	\pi \not\vDash C, \ \tau \not\vDash C_{\rest\pi} \big\} \\
&= \big\{ \tau \circ \pi(C) \setminus \{ 0 \} : C \in \Gamma, \
	\pi \not\vDash C, \ \tau \not\vDash C_{\rest\pi} \big\} \\
&= \big\{ C_{\rest \tau \circ \pi} : C \in \Gamma, \
	\tau \circ \pi \not \vDash C \big\}
= \Gamma_{\rest \tau \circ \pi}. \qedhere
\end{align*}
\end{proof}

A set of clauses $\Gamma$ \emph{semantically implies}
a clause~$C$, written $\Gamma\vDash C$,
if every total assignment satisfying $\Gamma$ also satisfies~$C$.
As is well-known, $\Gamma\vDash C$ holds if and only if there is
a \emph{resolution derivation} of some $C^\prime \subseteq C$; that is,
$C^\prime$ is derived from~$\Gamma$ using \emph{resolution inferences}
of the form
\begin{equation}\label{eq:resRule}
\AxiomC{$p \dotlor D$} 
\AxiomC{$\olnot p \dotlor E$} 
\BinaryInfC{$D \lor E$} 
\DisplayProof.
\end{equation}
If the derived clause $C^\prime$ is the empty clause~$\perp$, then the derivation is called
a \emph{resolution refutation} of~$\Gamma$.
By the soundness and completeness of resolution,
$\Gamma \vDash \perp$, that is, $\Gamma$ is unsatisfiable,
 if and only if
 there is a resolution refutation of~$\Gamma$.

If either $D$ or $E$ is empty, then the resolution inference~(\ref{eq:resRule})
is an instance of \emph{unit propagation}.
A refutation using only such inferences is called
a \emph{unit propagation refutation}.
Recall that we can write $\olnot C$
for the set of unit clauses $\{ \olnot p : p \in C \}$.

\begin{defi}
We write $\Gamma \vdash_1 \bot$ to denote that there is a unit
propagation refutation of $\Gamma$.
We define $\Gamma \vdash_1 C$ to mean $\Gamma \cup \olnot C \vdash_1 \bot$.
For a set of clauses $\Delta$, we write $\Gamma \vdash_1 \Delta$ to mean
$\Gamma \vdash_1 C$ for every~$C \in \Delta$.
\end{defi}

\begin{fact}\label{fac:unit_propagation}
If $\Gamma \vdash_1 \bot$ and $\alpha$ is any partial assignment
or substitution,
then $\Gamma_{\rest\alpha} \vdash_1 \bot$.
\end{fact}

In the literature, when $\Gamma \vdash_1 C$
then $C$ is said to be derivable from~$\Gamma$ by
\emph{reverse unit propagation} ($\tRUP$), or
is called an \emph{asymmetric tautology} (AT) with
respect to~$\Gamma$~\cite{VanGelder:RUP,JHB:inprocessing,HHW:verifying}.
Of course, $\Gamma \vdash_1 C$
implies that $\Gamma \vDash C$. The advantage of working with $\vdash_1$
is that
there is a simple polynomial time algorithm to
determine whether~$\Gamma \vdash_1 C$.
We have the following basic property of $\vdash_1$
(going back to~\cite{Chang:UnitAndInput}):

\begin{lem}\label{lem:RUP_vs_Res}
If $C$ is derivable from $\Gamma$ by a single resolution inference, then
$\Gamma \vdash_1 C$. Conversely, if~$\Gamma \vdash_1 C$,
then some $C^\prime \subseteq C$ has a resolution derivation from $\Gamma$ of
length at most~$n$, where $n$ is the total number of literals occurring in clauses
in~$\Gamma$.
\end{lem}

\begin{proof}
First suppose that $C = D \lor E$ and clauses $p \dotlor D$
and $\olnot{p} \dotlor E$ appear in $\Gamma$. Then by resolving these
with the unit clauses
in $\olnot{C}$ we can derive the two unit clauses $p$ and~$\olnot{p}$,
then resolve these together to get the empty clause.

Now suppose that $\Gamma \vdash_1 C$. Then there is a unit
propagation derivation of $\bot$ from $\Gamma \cup \olnot{C}$, which is
of length at most $n$. Removing all resolutions
against unit clauses~$\olnot{p}$ for $p \in C$,
this can be turned
into a resolution derivation of~$C$ or
of some $C^\prime\subseteq C$ from~$\Gamma$.
\end{proof}

\begin{lem}\label{lem:RUP_restriction}
Let $C\lor D$ be a clause (so $C\cup D$ is not tautological), and
set $\alpha = \olnot{C}$. Then
\[
\Gamma_{\rest\alpha} \vdash_1 D \setminus C
\quad \Longleftrightarrow \quad
\Gamma_{\rest\alpha} \vdash_1 D
\quad \Longleftrightarrow \quad
\Gamma \vdash_1 C \lor D.
\]
\end{lem}

\begin{proof}
The left-to-right directions are immediate from the definitions,
since $\Gamma_{\rest\alpha}$ is derivable from
$\Gamma \cup \alpha$ using unit propagation.
To show that
$\Gamma \vdash_1 C \lor D$ implies
$\Gamma_{\rest\alpha} \vdash_1 D \setminus C$,
suppose
$\Gamma \cup \alpha \cup \olnot{D} \vdash_1 \bot$
and apply Fact~\ref{fac:unit_propagation}.
\end{proof}

\subsection{Inference rules}%
\label{sec:inferences}

We will describe a series if inference rules
which can be used to add a clause~$C$ to a set of clauses~$\Gamma$.
In increasing order of strength the rules are
\[
\rm
\tBC \leftarrow
\tRAT \leftarrow
\tSPR \leftarrow
\tPR \leftarrow
\tSR .
\]
We will show that in each case the sets $\Gamma$ and $\Gamma \cup \{ C \}$ are \emph{equisatisfiable}, that is, either they are both
satisfiable or both unsatisfiable.
The definitions follow~\cite{JHB:inprocessing,HHW:verifying,HKB:StrongExtensionFree},
except for the new notion $\tSR$ of ``substitution redundancy''.\footnote{%
M.\ Heule [personal communication, 2018] has independently formulated an inference rule
``permutation redundancy'' ($\permR$) which allows only substitutions which set some variables
to constants and acts
as a permutation on the remaining literals.
This is a special case of $\tSR$; but unlike $\tSR$, $\permR$ does not allow
identifying distinct literals. However, we do not
know the strength of $\permRnnv$ relative to $\SRnnv$ (even if deletion
is allowed for both systems).}
All of these rules can be viewed as allowing the
introduction of clauses that hold ``without loss of
generality''~\cite{RebolaPardoSuda:SatPreserving}.
The rules are summarized in a table earlier in this section.

Let $\Gamma$ be a set of clauses and $C$ a
clause with a distinguished literal~$p$,
so that $C$ has the
form~$p \dotlor C^\prime$.

\begin{defiC}[\cite{Kullmann:WorstCase3SAT,Kullmann:NewMethods3SAT}]
The clause $C$ is a \emph{blocked clause} ($\tBC$)
with respect to $p$ and~$\Gamma$
if, for every clause $D$ of the form
$\olnot p \dotlor D^\prime$ in~$\Gamma$,
the set $C^\prime \cup D^\prime$ is tautological.
\end{defiC}

Notice that
the condition ``$C' \cup D'$ is tautological'' above
would be equivalent to $\emptyset \vdash_1 C' \vee D'$,
except that our notation does not allow us
to write the expression $C' \vee D'$ if $C' \cup D'$ is tautological,
since it is not a clause.
Since $\olnot p$ does not appear in $C'$ or~$D'$, it would also be
equivalent to $\emptyset \vdash_1\penalty10000 p \vee\penalty10000 C' \vee\penalty10000 D'$.
Compare with the
definition of RAT below.

\begin{defiC}[\cite{JHB:inprocessing,HeuleBiere:Variable,WHH:DRATtrim}]
A clause $C$ is a \emph{resolution asymmetric tautology} ($\tRAT$)
with respect to $p$ and~$\Gamma$
if, for every clause $D$ of the form
$\olnot p \dotlor D^\prime$ in~$\Gamma$, either
$C^\prime \cup D^\prime$ is
tautological or
\[
\Gamma \vdash_1 p \lor C^\prime \lor D^\prime.
\]
\end{defiC}

Here we write $p \lor C^\prime$ instead of $C$ to emphasize
that we include the literal~$p$ (some definitions of $\tRAT$ omit it).
Clearly, being $\tBC$ implies being $\tRAT$.

\begin{exaC}[\cite{Kullmann:GeneralizationER}]\label{ex:BC_extension}
Let $\Gamma$ be a set of clauses in which the
variable $x$ does not occur, but the variables $p$ and~$q$ may occur. Consider the three clauses
\begin{equation*} \label{eq:extensionClauses}
x\lor \olnot p \lor \olnot q
\quad\quad\quad
\olnot x \lor p
\quad\quad\quad
\olnot x \lor q
\end{equation*}
which together express that $x \leftrightarrow (p \wedge q)$.
Let $\Gamma_1 \subset \Gamma_2 \subset \Gamma_3$
be $\Gamma$ with the three clauses above successively added.
Then $x\lor \olnot p \lor \olnot q$ is $\tBC$ with respect to $\Gamma$ and $x$,
because no clause in $\Gamma$ contains $\olnot x$, so there is nothing
to check.
The second clause $\olnot x \lor p$ is $\tBC$ with respect to $\Gamma_1$ and $\olnot x$
because the only clause in $\Gamma_1$ containing $x$ is
$x\lor \olnot p \lor \olnot q$, and resolving this with $\olnot x \lor p$ gives a tautological conclusion.
The third clause $\olnot x \lor q$ is $\tBC$ with respect to $\Gamma_2$ and $\olnot x$ in a similar way.
\end{exaC}

It follows from the example that we can use the BC rule to
simulate extended resolution
if we are allowed to introduce new variables; see Section~\ref{sec:ER_new_variables}.

We say the clause $C$ is RAT with respect to~$\Gamma$
if it is RAT with respect to $p$ and~$\Gamma$ for some
literal~$p$ in~$C$, and similarly for~BC\@.

\begin{thmC}[\cite{Kullmann:GeneralizationER,JHB:inprocessing}]\label{thm:BCandRAT}
If $C$ is $\tBC$ or $\tRAT$ with respect to $\Gamma$, then
$\Gamma$ and $\Gamma \cup \{ C \}$ are equisatisfiable.
\end{thmC}

\begin{proof}
It suffices to show that if $\Gamma$ is satisfiable, then
so is $\Gamma\cup\{ C\}$.
Let $\tau$ be any total assignment satisfying $\Gamma$.
We may assume $\tau \vDash \olnot{C}$, as otherwise we are done.
Let $\tau^\prime$ be $\tau$ with the value of  $\tau(p)$
switched  to satisfy~$p$. Then $\tau^\prime$ satisfies $C$,
along with every clause in $\Gamma$ which does
not contain $\olnot{p}$.
Let $D = \olnot{p} \dotlor D^\prime$ be any clause in $\Gamma$ which
contains $\olnot{p}$. It follows from the $\tRAT$ assumption that
$\Gamma \vDash C \lor D^\prime$, so $\tau \vDash D^\prime$
since $\tau \vDash \olnot{C}$.
Hence $\tau^\prime \vDash D^\prime$ and thus $\tau^\prime \vDash D$.
This shows that~$\tau^\prime \vDash \Gamma \cup \{ C \}$.
\end{proof}

For the rest of this section, let $\alpha$ be the partial assignment $\olnot{C}$.
In a moment we will introduce the rules
SPR, PR and SR\@. These are variants of a common form, and
we begin by showing that RAT can also be expressed in
a similar way (in the literature this form of RAT
is called \emph{literal propagation redundant} or LPR).

\begin{thmC}[\cite{HKB:StrongExtensionFree}]\label{thm:RATandLPR}
A clause $C$ is $\tRAT$ with respect to $p$ and~$\Gamma$
if and only if
$\Gamma_{\rest \alpha} \vdash_1 \Gamma_{\rest \tau}$
where $\tau$ is the partial assignment identical to
 $\alpha$ except at $p$, with $\tau(p)=1$.
\end{thmC}

\begin{proof}
First suppose that $C$ satifies the second condition.
Consider any clause~$D$ of the form
$\olnot p \dotlor D^\prime$ in~$\Gamma$.
We need to show that either $C \cup D^\prime$ is tautological or
$\Gamma \vdash_1 C \lor D^\prime$.
Suppose $C \cup D^\prime$ is not tautological.
Then $\alpha \not \vDash D^\prime$, $\tau \not \vDash D$, and
 by Lemma~\ref{lem:RUP_restriction}
it is enough to show $\Gamma_{\rest\alpha} \vdash_1 D^\prime$.
But this now follows
from 
$\Gamma_{\rest\alpha} \vdash_1 D_{\rest \tau}$,
since $D_{\rest \tau} = {D^\prime}_{\rest \alpha} \subseteq D^\prime$.

Now suppose $C$ is $\tRAT$ with respect to $p$ and~$\Gamma$.
 Consider any  $D \in \Gamma$ such that $\tau \not \vDash D$
and thus~$D_{\rest\tau} \in \Gamma_{\rest \tau}$.
We must show that $\Gamma_{\rest\alpha} \vdash_1 D_{\rest\tau}$.
If $\olnot{p} \notin D$ this is trivial, since then
$D_{\rest\tau} = D_{\rest\alpha} \in \Gamma_{\rest\alpha}$.
Otherwise $D = \olnot{p} \dotlor D^\prime$, where
$\alpha \not \vDash D^\prime$ since $\tau \not \vDash D$,
so $C \cup D^\prime$ is not tautological.
By the $\tRAT$ property, $\Gamma \vdash_1 C \lor D^\prime$.
By Lemma~\ref{lem:RUP_restriction} this implies
$\Gamma_{\rest\alpha} \vdash_1 D^\prime \setminus C$.
But $D^\prime \setminus C = {D^\prime}_{\rest\alpha} = D_{\rest\tau}$.
\end{proof}

\begin{defiC}[\cite{HKB:StrongExtensionFree}]\label{def:SPR}
A clause $C$ is \emph{subset propagation redundant} ($\tSPR$) with respect to~$\Gamma$
if there is a partial assignment~$\tau$ with $\tdom(\tau)=\tdom(\alpha)$ such that $\tau \vDash C$ and
$\Gamma_{\rest \alpha} \vdash_1 \Gamma_{\rest\tau}$.
\end{defiC}

\begin{defiC}[\cite{HKB:StrongExtensionFree}]\label{def:PR}
A clause $C$ is \emph{propagation redundant} ($\tPR$) with respect to~$\Gamma$
if there is a partial assignment~$\tau$ such that $\tau \vDash C$ and
$
\Gamma_{\rest \alpha} \vdash_1 \Gamma_{\rest\tau}.
$
\end{defiC}

\begin{defi}\label{def:SR}
A clause $C$ is \emph{substitution redundant} ($\tSR$) with respect to~$\Gamma$
if there is a substitution~$\tau$ such that $\tau \vDash C$ and
$\Gamma_{\rest \alpha} \vdash_1 \Gamma_{\rest\tau}$.
\end{defi}

\begin{exaC}[based on~\cite{HKB:StrongExtensionFree}]\label{ex:PHP_in_SR}
Let $\Gamma$ be the pigeonhole principle $\PHP_n$ (see Section~\ref{sec:PHP})
in variables~$p_{i,j}$ expressing that pigeon~$i$
goes to hole~$j$.
Let $C$ be the clause $\olnot{p_{1,0}} \vee p_{0,0}$
so that $\alpha$
is the partial assignment $p_{1,0} \wedge \olnot{p_{0,0}}$.

Let $\pi$ be the substitution
which swaps pigeons $0$ and~$1$; that is,
$\pi(p_{0,j}) = p_{1,j}$ and $\pi(p_{1,j}) = p_{0,j}$
for every hole $j$, and $\pi$ is otherwise the identity.
Notice that, by the symmetries of the pigeonhole principle,
$\Gamma_{\rest \pi} = \Gamma$ and thus
$\Gamma_{\rest\alpha} = {(\Gamma_{\rest\pi})}_{\rest\alpha}
=\Gamma_{\rest \alpha \circ \pi}$.
Let $\tau = \alpha \circ \pi$, so $\tau$ is the same as $\pi$
except that $\tau(p_{0,0})=1$ and $\tau(p_{1,0}) = 0$.

Then $\tau \vDash C$ and $\Gamma_{\rest \alpha} \vdash_1 \Gamma_{\rest\tau}$ (since they are the same set of clauses).
Hence we have shown that~$C$ is $\tSR$ with respect to $\Gamma$.

We go on to sketch a polynomial size
$\DSRnnv$ refutation of $\Gamma$,
that is, one that uses $\tSR$ inferences, resolution and deletion
but introduces no new variables
(see Section~\ref{sec:alternateNNV} below).
Resolve $C$ with the hole axiom
$\olnot{p_{1,0}} \vee \olnot{p_{0,0}}$ to derive
the unit clause~$\olnot{p_{1,0}}$. Delete $C$,
so that we are now working with the set of clauses
\mbox{$\Gamma \cup \{ \olnot{p_{1,0}} \}$}.
Let $C'$ be the clause $\olnot{p_{2,0}} \vee p_{0,0}$
and let $\alpha'$
be its negation~\mbox{$p_{2,0} \wedge \olnot{p_{0,0}}$}.
Let $\pi'$ be the substitution
which swaps pigeons~$0$ and~$2$
and let $\tau' = \alpha' \circ \pi'$.
As before $\tau' \vDash C'$ and
${(\Gamma \cup \{ \olnot{p_{1,0}} \})}_{\rest\alpha'}
\vdash_1 {(\Gamma \cup \{ \olnot{p_{1,0}} \})}_{\rest\tau'}$,
since neither $\alpha'$ nor $\tau'$ affects $p_{1,0}$
so these are again the same set of clauses.
Hence we may derive $C'$ by a $\tSR$ inference, then
resolve with the hole axiom
$\olnot{p_{2,0}} \vee \olnot{p_{0,0}}$
to get $\olnot{p_{2,0}}$.

Carrying on in this way, we eventually derive
$\Gamma \cup \{ \olnot{p_{1,0}} \}
\cup \dots \cup \{ \olnot{p_{n-1,0}} \}$.
We now resolve each unit clause $\olnot{p_{i,0}}$
with the pigeon axiom for pigeon $i$,
for $i=1, \dots, n-1$.
After some deletions, we are left
with clauses asserting that pigeons $1, \dots, n-1$
map injectively to holes $1, \dots, n-2$.
This is essentially~$\PHP_{n-1}$. We carry on inductively
to derive $\PHP_{n-2}$ etc.\ and can easily derive a contradiction
when we get to $\PHP_2$.
\end{exaC}

Section~\ref{sec:PHP} contains a more careful
version of this argument, refuting $\PHP_n$
using $\tSPR$ inferences and no
deletion.

\begin{thm}\label{thm:SRequisat}
If $C$ is $\tSR$ with respect to $\Gamma$, then $\Gamma$ and $\Gamma\cup\{C\}$ are equisatisfiable.
Hence the same is true for $\tSPR$ and $\tPR$.
\end{thm}

\noindent
Theorem~\ref{thm:SRequisat} trivially implies the same statement
for $\tBC$ and $\tRAT$ (this was Theorem~\ref{thm:BCandRAT}
above).

\begin{proof}
Again it is sufficient to show that if $\Gamma$ is satisfiable,
then so is $\Gamma \cup \{ C \}$.
Suppose we have a substitution~$\tau$ such that $\tau \vDash C$ and
$\Gamma_{\rest \alpha} \vdash_1 \Gamma_{\rest\tau}$.
Let $\pi$ be any total assignment satisfying $\Gamma$.
If $\pi \vDash C$ then
we are done. Otherwise $\pi \vDash \olnot{C}$, so
$\pi \supseteq \alpha$ and $\pi$~satisfies $\Gamma_{\rest\alpha}$ by Lemma~\ref{lem:substitution_composition}.
Thus, by the assumption, $\pi \vDash \Gamma_{\rest\tau}$.
Therefore $\pi\circ\tau \vDash \Gamma$ by Lemma~\ref{lem:substitution_composition},
and $\pi \circ \tau \vDash C$ since $\tau \vDash C$.
\end{proof}

This proof of Theorem~\ref{thm:SRequisat}
 still goes through if we replace
$\Gamma_{\rest \alpha} \vdash_1 \Gamma_{\rest\tau}$
 with the weaker assumption
that~$\Gamma_{\rest \alpha} \vDash \Gamma_{\rest \tau}$.
The advantage of using $\vdash_1$ is that it is efficiently checkable.
Consequently, the conditions of being BC, RAT, SPR, PR or SR with respect
to $\Gamma$ are all polynomial-time checkable, as long
as we include the partial assignment or substitution $\tau$
as part of the input.



\subsection{Proof systems with new variables}%
\label{sec:proofsystems}

This section introduces proof systems based on
the $\tBC$, $\tRAT$,  $\tSPR$, $\tPR$ and $\tSR$ inferences.
Some of the systems also allow the use of the
deletion rule: these systems are denoted
$\tDBC$, $\tDRAT$,  etc. All the proof systems
are \emph{refutation systems}. They
start with a set of clauses $\Gamma$, and successively
derive sets~$\Gamma_i$ of clauses, first $\Gamma_0 = \Gamma$,
then $\Gamma_1,\Gamma_2, \ldots, \Gamma_m$ until
reaching a set $\Gamma_m$ containing the empty clause.
It will always be the case that if $\Gamma_i$ is satisfiable,
then $\Gamma_{i+1}$ is satisfiable. Since the empty clause~$\perp$
is in $\Gamma_{m}$, this last set is not satisfiable. This
suffices to show that $\Gamma$ is not satisfiable.

\begin{defi}%
\label{def:proofsystems}
A $\tBC$, $\tRAT$, $\tSPR$, $\tPR$ or $\tSR$ 
proof (a refutation)
of $\Gamma$ is a sequence $\Gamma_0, \ldots, \Gamma_m$
such that $\Gamma_0 = \Gamma$, $\perp\in\Gamma_m$ and
each $\Gamma_{i+1} =\Gamma_i\cup\{C\}$, where either
\begin{itemize}
\item
$\Gamma_i \vdash_1 C$ 
 or
\item
$C$ is $\tBC$, $\tRAT$, $\tSPR$, $\tPR$, or $\tSR$ (respectively) with respect to $\Gamma_i$.
\end{itemize}
For $\tBC$ or $\tRAT$ steps, the proof must specify which~$p$ is used,
 and for $\tSPR$, $\tPR$ or $\tSR$, it must specify which~$\tau$.
\end{defi}

There is no constraint on the variables
that appear in clauses $C$ introduced in $\tBC$, $\tRAT$, etc.\ steps.
They are free to include new variables that did not occur
in~$\Gamma_0, \dots, \Gamma_i$.

\begin{defi}\label{def:proofsystemsDeletion}
A $\tDBC$, $\tDRAT$, $\tDSPR$, $\tDPR$ or $\tDSR$ proof allows
the same rules of
inference (respectively) as Definition~\ref{def:proofsystems},
plus the \emph{deletion} inference rule:
\begin{itemize}
\item
$\Gamma_{i+1} = \Gamma_i \setminus \{ C \}$
for some $C\in\Gamma_i$.
\end{itemize}
\end{defi}


\noindent
Resolution can be simulated by $\tRUP$ inferences (Lemma~\ref{lem:RUP_vs_Res}),
so all the
systems introduced in this and the next subsection
simulate resolution. Furthermore, by Theorems~\ref{thm:BCandRAT} and~\ref{thm:SRequisat},
they are sound. Since the inferences are defined
using $\vdash_1$, they are polynomial time verifiable, as the
description of~$\tau$ is included with every $\tSPR$, $\tPR$ or $\tSR$ inference. Hence they
are all proof systems in the sense of
Cook-Reckhow~\cite{CookReckhow:proofsstoc,CookReckhow:proofs}.

The deletion rule deserves more explanation. First, we allow
\emph{any} clause to be deleted, even the initial clauses from $\Gamma$. So it is possible that $\Gamma_i$ is
unsatisfiable but
$\Gamma_{i+1}$ is satisfiable after a deletion. For us, this is okay
since we focus on refuting sets of unsatisfiable clauses, not on
finding satisfying assignments of satisfiable sets of clauses.
SAT solvers generally wish to maintain the equisatisfiability
property: they use deletion extensively to prune the search time,
but are careful only to perform deletions that preserve both
satisfiability and unsatisfiability, generally as justified by
the $\tBC$ or $\tRAT$ rules. Since applying $\tRAT$, or more generally $\tPR$ or $\tSR$, can
change the satisfying assignment, a SAT solver may also need to keep
a proof log with information about how to reverse the steps of
the proof once a satisfying assignment is found
(see~\cite{JHB:inprocessing}).

Second, deletion is important for us because the property of
being $\tBC$, $\tRAT$ etc.\ involves a universal quantification over the current
set of clauses $\Gamma_i$. So deletion can make the systems
more powerful, as removing clauses from $\Gamma_i$ can make more inferences possible.
For example, the unit clause $x$ is BC with respect to
the set~\mbox{$\{ x \lor y \}$}, since the literal $\olnot{x}$
does not appear, but is not even SR with respect to the
set~$\{ x \lor y, \, \olnot{x} \}$,
since~$\{ x \lor y, \, \olnot{x} \}$
and~$\{ x \lor y, \, \olnot{x}, \, x \}$ are not equisatisfiable.
An early paper on this
by Kullmann~\cite{Kullmann:GeneralizationER} exploited
deletions to generalize the power
of $\tBC$ inferences.

As we will show in Section~\ref{sec:extended_resolution}, all the systems defined so
far are equivalent to extended resolution, because of the ability to freely introduce new variables.
The main topic of the paper is the systems we introduce next,
which lack this ability.

\subsection{Proof systems without new variables}\label{sec:alternateNNV}

\begin{defi}\label{def:noNewVars}
A $\tBC$ refutation of $\Gamma$ \emph{without new variables},
or, for short, a $\BCnnv$ refutation of $\Gamma$,
is a $\tBC$ refutation of $\Gamma$ in which only variables from $\Gamma$ appear.
The systems $\RATnnv$, $\PRnnv$ etc.\ and $\DBCnnv$, $\DRATnnv$, $\DPRnnv$ etc.\ are
defined similarly.
\end{defi}

There is an alternative natural definition of ``without new variables'',
which requires not just that a refutation of~$\Gamma$
uses only variables that are used in~$\Gamma$, but also that once a variable
has been eliminated from all clauses through the use of deletion, it may not be
reused subsequently in the refutation. An equivalent way to state this
is that a clause~$C$ inferred
by a $\tBC$, $\tRAT$, $\tSPR$, $\tPR$ or $\tSR$ inference cannot involve
any variable which does not occur in the \emph{current} set of clauses.

This stronger definition is in fact essentially
equivalent to Definition~\ref{def:noNewVars},  for
a somewhat trivial reason. More precisely, any refutation
that satisfies Definition~\ref{def:noNewVars} can be converted into
a refutation that satisfies the stronger condition with at worst a
polynomial increase in the size of the refutation.
We state the proof for $\DBCnnv$, but the same argument works
verbatim for the other systems $\DRATnnv$, $\DSPRnnv$, $\DPRnnv$ and $\DSRnnv$.

Suppose $\Pi$ is a $\DBCnnv$ refutation of $\Gamma$ in the sense
of Definition~\ref{def:proofsystemsDeletion},
and consider a variable~$x$.
Suppose~$x$ is present in
$\Gamma = \Gamma_0$ and in $\Gamma_i$, is not present in~$\Gamma_{i+1}$ through~$\Gamma_j$,
but is present again in~$\Gamma_{j+1}$.  The derivation of $\Gamma_{i+1}$ from
$\Gamma_i$ deleted a single clause~$x \lor C$; for definiteness we assume this clause contains~$x$
positively. The derivation of $\Gamma_{j+1}$ introduced a clause $x \lor D$
with a $\tBC$ inference; we may assume without loss of generality that $x$~occurs with the same
sign in~$x \lor D$ as in~\mbox{$x \lor C$}, since otherwise the sign of~$x$
could be changed throughout the refutation from~$\Gamma_{j+1}$ onwards.

The refutation~$\Pi$ is modified as follows. Before deleting the clause $x \lor C$,
infer the unit clause $x$ by a $\tBC$ inference; this is valid trivially,
since $\olnot x$ does not occur in~$\Gamma_i$. Then
continue the derivation with the unit clause~$x$ added to $\Gamma_i,\ldots, \Gamma_j$.
Since there are no other uses of~$x$ in $\Gamma_i,\ldots, \Gamma_j$, these steps
in the refutation remain valid (by part~(b) of Lemma~\ref{lem:subsumeNoDelete} below).
Upon reaching $\Gamma_j$, infer $x\lor D$ with a $\tBC$ inference relative to the variable~$x$.
This is allowed since $\olnot x$ does not appear in $\Gamma_j$.  Then delete the
unit clause~$x$ to obtain again $\Gamma_{j+1}$.
Repeating this  for every gap in $\Pi$ where $x$ disappears,
and then doing the same construction for every variable, yields a $\DBCnnv$ refutation that satisfies the
stronger condition.


\subsection{Two useful lemmas}

We conclude this subsection with two technical lemmas, which
we will use in several places to simplify the construction of proofs.

All the inference rules $\tBC$, $\tRAT$, $\tSPR$, $\tPR$ and $\tSR$ are ``non-monotone'', in the sense that
it is possible that $\Gamma_{\rest \alpha}\vdash_1 \Gamma_{\rest\tau}$ holds
but $\Gamma^\prime_{\rest \alpha}\vdash_1 \Gamma^\prime_{\rest\tau}$ fails,
for $\Gamma \subseteq \Gamma^\prime$.
In particular, adding more clauses to~$\Gamma$ may invalidate
a $\tBC$, $\tRAT$, $\tSPR$, $\tPR$ or $\tSR$ inference. Conversely, removing clauses
from~$\Gamma$ may allow new clauses to be inferred by one of these
inferences.
This is one reason for the importance of the deletion rule.

The next lemma is a useful technical tool that will sometimes let us
avoid using deletion. It states conditions under which
the extra clauses in~$\Gamma^\prime$ do not invalidate
a $\tRAT$, $\tSPR$, $\tPR$ or $\tSR$ inference.\footnote{The conclusion
of Lemma~\ref{lem:subsumeNoDelete} is true
also for $\tBC$ inferences.}

\begin{defi}
A clause $C$ \emph{subsumes} a clause~$D$ if $C\subseteq D$.
A set $\Gamma$ of clauses \emph{subsumes} a set~$\Gamma^\prime$ if
each clause of $\Gamma^\prime$ is subsumed by some clause of~$\Gamma$.
\end{defi}


\begin{lem}\label{lem:subsumeNoDelete}
Suppose $\alpha$ and $\tau$ are substitutions and
$\Gamma_{\rest \alpha} \vdash_1 \Gamma_{\rest \tau}$ holds.
Also suppose $\Gamma\subseteq\Gamma^\prime$.
\begin{enumerate}[(a)]
\item
If $\Gamma$ subsumes~$\Gamma^\prime$,
then $\Gamma^\prime_{\rest \alpha} \vdash_1 \Gamma^\prime_{\rest \tau}$.
\item
If $\Gamma^\prime$ is $\Gamma$ plus one or more clauses
involving only variables that are not in the domain of either $\alpha$ or~$\tau$,
then $\Gamma^\prime_{\rest \alpha} \vdash_1 \Gamma^\prime_{\rest \tau}$.
\end{enumerate}
Consequently, in either case, if~$C$ can be inferred from~$\Gamma$ by
a $\tRAT$, $\tSPR$, $\tPR$ or $\tSR$ rule, then $C$~can also be inferred from~$\Gamma^\prime$
by the same rule.
\end{lem}

\begin{proof}
We prove (a). Suppose $D \in \Gamma^\prime$ and $\tau \not \vDash D$.
We must show $\Gamma^\prime_{\rest \alpha} \vdash_1 D_{\rest\tau}$.
Let $E \in \Gamma$ with~$E \subseteq D$. Then
$\tau \not \vDash E$, so by assumption
$\Gamma_{\rest\alpha} \vdash_1 E_{\rest\tau}$.
Also $E_{\rest\tau} \subseteq D_{\rest\tau}$, so
$\Gamma_{\rest \alpha} \vdash_1 D_{\rest\tau}$.
It follows that~$\Gamma^\prime_{\rest \alpha} \vdash_1 D_{\rest\tau}$, since $\Gamma \subseteq \Gamma^\prime$.

The proof of (b) is immediate from the definitions.
\end{proof}

Our last lemma
gives a kind of normal form for propagation
redundancy. Namely, it
implies that when $C$ is $\tPR$ with respect to~$\Gamma$,  we
may assume without loss of generality that~$\tdom(\tau)$
includes~$\tdom(\alpha)$. We will use this later to show a limited simulation of $\tPR$ by~$\tSPR$.

\begin{lem}\label{lem:PRnormal}
Suppose $C$ is $\tPR$ with respect to $\Gamma$, witnessed by a partial assignment~$\tau$.
Then we have~$\Gamma_{\rest\alpha} \vdash_1 \Gamma_{\rest\alpha\circ\tau}$,
where $\alpha$ is the partial assignment $\olnot{C}$.
\end{lem}

\begin{proof}
Let $\pi = \alpha \circ \tau$. Suppose $E \in \Gamma$ is such that
$\pi \not \vDash E$.
 We must show that $\Gamma_{\rest\alpha} \vdash_1 E_{\rest\pi}$.
We can decompose $E$ as $E_1 \lor E_2 \lor E_3$ where $E_1$ contains the literals
in $\tdom(\tau)$, $E_2$ contains the literals in $\tdom(\alpha) \setminus \tdom(\tau)$
and $E_3$ contains the remaining literals.
Then $E_{\rest\tau} = E_2 \lor E_3$ and by the $\tPR$
assumption~\mbox{$\Gamma_{\rest\alpha} \vdash_1 E_{\rest\tau}$}, so there is a derivation
$\Gamma_{\rest\alpha} \cup \olnot{E_2} \cup \olnot{E_3} \vdash_1 \bot$.
But neither $\Gamma_{\rest\alpha}$ nor $\olnot{E_3}$ contain any variables from $\tdom(\alpha)$,
so the literals in $\olnot{E_2}$ are not used in this derivation.
Hence  $\Gamma_{\rest\alpha}  \cup \olnot{E_3} \vdash_1 \bot$,
which completes the proof since $E_3=E_{\rest\pi}$.
\end{proof}

\section{Relations with extended resolution}\label{sec:extended_resolution}
\subsection{With new variables}\label{sec:ER_new_variables}

It is known that $\tRAT$, and even $\tBC$,
can simulate extended resolution if new variables are
allowed~\cite{Kullmann:GeneralizationER}.
In extended resolution for any variables $p,q$
we are allowed to introduce a new variable $x$
together with three clauses expressing
that $x \leftrightarrow (p \wedge q)$.
As shown in Example~\ref{ex:BC_extension},
we can instead introduce these clauses using $\tBC$
inferences.
Thus all the systems described above which allow new variables
simulate extended resolution.
The converse holds as well:

\begin{thm}\label{thm:ERsimulates}
The system $\tER$ simulates $\tDSR$, and hence every other system above.
\end{thm}

\begin{proof} (Sketch)
It is known that
the theorem holds for $\tDPR$ in place of $\tDSR$.
Namely, \cite{KRPH:erDRAT} gives an explicit simulation 
of $\tDRAT$ by extended resolution, and~\cite{HeuleBiere:Variable} gives an explicit
simulation of $\tDPR$ by $\tDRAT$.
Thus extended
resolution simulates $\tDPR$.

We sketch a direct proof of the
simulation of $\tDSR$ by extended resolution.
Suppose $\Gamma_0, \ldots, \Gamma_m$ is a DSR proof
and in particular $\Gamma_{i+1} = \Gamma_i \cup \{ C \}$
is introduced by an $\tSR$ inference from~$\Gamma_i$ with a substitution~$\tau$.
Let $x_1, \dots, x_s$ be all variables occurring in $\Gamma_{i+1}$
including any new variables introduced in~$C$.
Using the extension rule,
introduce new variables $x_1^\prime, \ldots, x_s^\prime$
along with extension variables and extension axioms expressing
\[
x_j^\prime ~\liff~  ( x_j \land C) \lor (\tau(x_j)\land \lnot C).
\]
Here $\tau(x_j)$ stands for a fixed symbol from $\tLit \cup \{ 0,1\}$.
Let $\Gamma_{i+1}(\vec x/\vec x^\prime)$ be the set of clauses
obtained from $\Gamma_i$ by replacing each variable $x_j$
with $x_j^\prime$.
It can be proved using only resolution, using the extension axioms, that if all
 clauses in~$\Gamma_i$ hold then all clauses
in $\Gamma_{i+1}(\vec x/\vec x^\prime)$ hold.  The extended
resolution proof then proceeds inductively on~$i$
using the new variables~$x_j^\prime$ in place of the
old variables~$x_j$.

Another way to prove the full theorem is via the
theories of bounded arithmetic $S^1_2$~\cite{Buss:bookBA}
and PV~\cite{Cook:PV}. 
 Namely,
suppose we are given a $\tDSR$ proof $\Gamma_0,\ldots,\Gamma_m$
and a satisfying assignment $\pi_0$
for $\Gamma_0$. By induction,
there exists a satisfying assignment~$\pi_i$ for
each~$\Gamma_i$.
The inductive step, for an  $\tSR$ rule deriving $\Gamma_{i+1} = \Gamma_i \cup \{C\}$,
witnessed by a substitution~$\tau$,
is to set~$\pi_{i+1} = \pi_i$ if $\pi_i \vDash C$ and otherwise
to set~$\pi_{i+1} = \pi_i \circ \tau$,
as in the proof of Theorem~\ref{thm:SRequisat}.
The inductive hypothesis can be written as a $\Sigma^b_1$ formulas
and the induction has $m$ steps, so this is formally
$\Sigma^b_1$ length-induction
(or $\Sigma^b_1$-LIND) which is available in $S^1_2$.
Thus $S^1_2$~can prove the soundness
of $\tDSR$. By conservativity of $S^1_2$ over PV~\cite{Buss:bookBA},
the theory PV also proves the soundness of $\tDSR$.
A fundamental property of PV is that PV proofs
translate into uniform families of polynomial size extended resolution refutations~\cite{Cook:PV}.
Thus ER efficiently proves the soundness of $\tDSR$.
It follows that ER simulates $\tDSR$.
\end{proof}

\subsection{Without new variables}\label{sec:ER_nnv}

In the systems without the ability to freely add new variables,
we can still imitate extended resolution by adding dummy variables to the formula we want to refute.
This was observed already in~\cite{Kullmann:GeneralizationER}.

For $m \ge 1$, define $X^m$ to be the set consisting of only the
two clauses
\[
y \lor x_1 \lor \dots \lor x_m
\qquad \hbox{and} \qquad
y.
\]

\begin{lem}\label{lem:ER_to_BC}
Suppose $\Gamma$ has an $\tER$ refutation $\Pi$ of size $m$,
and that $\Gamma$ and $X^m$ have no variables in common.
Then $\Gamma \cup X^m$ has a $\BCnnv$-refutation~$\Pi^*$ of size
$O(m)$, which can furthermore
be constructed from $\Pi$ in polynomial time.
\end{lem}

\begin{proof}
We describe how to change $\Pi$ into $\Pi^*$. We first
rename all extension variables to use names from $\{x_1, \dots, x_m\}$
and replace all resolution steps with $\vdash_1$ inferences.
Now consider an extension rule in~$\Pi$
which introduces the three extension
clauses~(\ref{eq:extensionClauses}) expressing
$x_i \ \liff \ (p \land q)$,
where we may assume that $p$ and $q$ are either variables of $\Gamma$
or from~$\{x_1, \dots, x_{i-1}\}$.
We simulate this by introducing successively the
three clauses
\[
x_i \lor \olnot p \lor \olnot q
\quad\quad\quad
\olnot x_i \lor p \lor \olnot y
\quad\quad\quad
\olnot x_i \lor q \lor \olnot y
\]
using the $\tBC$ rule.
The first clause, $x_i \lor \olnot p \lor \olnot q$, is $\tBC$ with respect to~$x_i$,
because $\olnot x_i$ has not appeared yet.
The second clause is $\tBC$ with respect to $\olnot{x_i}$,
because $x_i$ appears only in two earlier clauses, namely
$y \lor x_1 \lor \dots \lor x_m$,
which contains $y$, and
$x_i \lor \olnot p \lor \olnot q$,
which contains $\olnot p$. In both cases
the resolvent with $\olnot x_i \lor p \lor \olnot y$
is tautological.
 The third clause is similar.
The unit clause~$y$ is in $X^m$, so we can then
derive the remaining two needed extension clauses
$\olnot x_i \lor p$ and $\olnot x_i \lor q$
by two $\vdash_1$ inferences.
\end{proof}

In the terminology of~\cite{PitassiSanthanam:effectivePsim},
Lemma~\ref{lem:ER_to_BC} shows that
$\BCnnv$ \emph{effectively simulates} $\tER$,
in that we are allowed to transform the formula as well as the refutation
when we move from $\tER$ to $\BCnnv$.

The next corollary is essentially from~\cite{Kullmann:GeneralizationER}.
It shows how to use the lemma to construct examples
of usually-hard formulas which have short proofs in $\BCnnv$.
(We will give less artificial examples
of short $\SPRnnv$ proofs in Section~\ref{sec:Upperbounds}.)
Let $m(n)$ be the polynomial size upper bound on $\tER$ refutations
of the pigeonhole principle $\PHP_n$
which follows from~\cite{CookReckhow:proofs} --- see Section~\ref{sec:PHP}
for the definition of the $\PHP_n$ clauses.

\begin{cor}\label{cor:PHP_X}
The set of clauses $\PHP_n \cup X^{m(n)}$ has polynomial size
proofs in $\BCnnv$, but requires exponential size proofs
in constant depth Frege.
\end{cor}

\begin{proof}
The upper bound follows from Lemma~\ref{lem:ER_to_BC}.
For the lower bound, let $\Pi$ be a refutation
in depth-$d$ Frege. Then we can restrict $\Pi$
by setting $y=1$ to obtain a depth-$d$
refutation of $\PHP_n$ of the same size.
By~\cite{KPW:PHP,PBI:PHP}, this must have
exponential size.
\end{proof}

The same argument can give a more general
result.
A propositional proof system $\calP$ is
\emph{closed under restrictions} if,
given any $\calP$-refutation of $\Gamma$ and any
partial assignment~$\rho$,
we can construct a $\calP$-refutation of $\Gamma_{\rest\rho}$
in polynomial time. (\cite{BKS:clauselearning} called such systems ``natural''.)
Most of the commonly-studied proof systems such as resolution,
Frege, etc.~are closed under restrictions.
On the other hand, it follows from  results in this paper
that $\BCnnv$ and $\RATnnv$ are not closed under restrictions.
To see this,
let~$\Gamma$ be $\BPHP_n \cup X^{m(n)}$ where $\BPHP_n$ is
the bit pigeonhole principle (see Section~\ref{sec:BPHP})
and $m$ is a suitable function.
Then $\Gamma$ has short $\BCnnv$ refutations,
since $\BPHP_n$ has short refutations in~$\tER$. But $\BPHP_n$ is a restriction of~$\Gamma$, as in Corollary~\ref{cor:PHP_X}, and
has no short $\RATnnv$ refutations by Theorem~\ref{thm:BPHP_size} below.

\begin{thm}
Let $\calP$ be any propositional proof system which is closed
under restrictions. If $\calP$ simulates $\BCnnv$, then $\calP$
simulates $\tER$.
\end{thm}

\begin{proof}
Suppose $\Gamma$ has a refutation $\Pi$ in $\tER$ of length~$m$.
Take a copy of~$X^m$ in disjoint variables from~$\Gamma$.
By Lemma~\ref{lem:ER_to_BC} we can construct
a $\BCnnv$-refutation  of~\mbox{$\Gamma \cup X^m$}.
Since $\calP$ simulates $\BCnnv$, we can then construct
a $\calP$-refutation of~\mbox{$\Gamma \cup X^m$}.
Let $\rho$ be the restriction which just sets $y=1$,
so that~${(\Gamma \cup X^m)}_{\rest\rho} = \Gamma$.
By the assumption that $\calP$ is closed under restrictions,
we can construct a $\calP$ refutation of~$\Gamma$.
All constructions are polynomial time.
\end{proof}

\begin{cor}
If the Frege proof system simulates $\BCnnv$, then
Frege and $\tER$ are equivalent.
\end{cor}

Hence it is unlikely that Frege simulates $\BCnnv$, since
Frege is expected to be strictly weaker than~$\tER$.

\subsection{Canonical NP pairs}

The notion of \emph{disjoint NP pairs} was first introduced
by Grollmann and Selman~\cite{GrollmannSelman:CryptoMeasures}.
Razborov~\cite{Razborov:NPpairs} showed how a propositional proof system~$\calP$
gives rise to a canonical disjoint NP pair,
which gives a measure of the strength of the system.
It is known that
if a propositional proof system~$\calP_1$ simulates a system~$\calP_2$,
then there is a many-one reduction from the canonical NP pair for~$\calP_2$
to the canonical NP pair for~$\calP_1$~\mbox{\cite{Razborov:NPpairs,Pudlak:NPpairs}}.
We can use Lemma~\ref{lem:ER_to_BC} to prove that the
systems $\BCnnv$ through $\DSRnnv$ cannot be distinguished
from each other or $\tER$ by their canonical NP pairs,
even though they do not all simulate each other.

\begin{defi}
A \emph{disjoint NP pair} is a pair $(U,V)$ of NP sets such that
$U \cap V = \emptyset$.
A \emph{many-one reduction} from a disjoint NP pair $(U,V)$ to
a disjoint NP pair $(U^\prime, V^\prime)$ is a polynomial time function~$f$
mapping $U$ to~$U^\prime$ and mapping $V$ to~$V^\prime$.
\end{defi}

To motivate this definition a little,
a disjoint NP pair $(U,V)$ is said to be \emph{polynomially separable} if
there is a polynomial time function~$f$ which, given
\mbox{$x \in U \cup V$}, correctly
identifies whether $x\in U$ or $x\in V$.
Clearly if $(U,V)$ is many-one reducible to $(U^\prime, V^\prime)$,
then if $(U^\prime,V^\prime)$ is polynomially separable so is $(U,V)$.

\begin{defi}
$\SAT$ is the set of pairs $(\Gamma, 1^m)$ such that $\Gamma$
is a satisfiable set of clauses and $m \ge 1$~is an arbitrary integer.
Let $\calP$ be a propositional proof system for refuting sets of clauses.
Then $\REF(\calP)$ is the set of pairs $(\Gamma, 1^m)$ such that
$\Gamma$~has a $\calP$-refutation of length at most $m$.
Notice that $\SAT$ and $\REF(\calP)$ are both NP\@.
We define
the \emph{canonical disjoint NP pair}, or \emph{canonical NP pair},
associated with~$\calP$ to be $(\REF(\calP),\SAT)$.
\end{defi}

The canonical NP pair for a proof system~$\calP$ defines the following
problem. Given a pair $(\Gamma,1^m)$, the soundness of~$\calP$
implies that it is impossible that both (a)~$\Gamma$~is satisfiable
and (b)~$\Gamma$ has a proof in~$\calP$ of length~$\le m$.  The promise
problem is to identify one of (a) and~(b) which does \emph{not} hold.
(If neither (a) nor~(b) holds, then either answer may be given.)

\begin{thm}\label{thm:BCandERpairs}
There are many-one reductions in both directions between
the canonical NP pair for $\tER$
and the
canonical NP pairs for
all the systems in Section~\ref{sec:proofsystems}.
\end{thm}

\begin{proof}
As a simulation implies a reduction between canonical NP pairs,
all we need to show is a reduction of the canonical NP pair for $\tER$
to the canonical NP pair for the weakest system $\BCnnv$,
that is, of
$(\REF(\tER),\SAT)$ to $(\REF(\BCnnv),\SAT)$.
Suppose $(\Gamma,1^m)$ is given as a query to $(\REF(\tER),\SAT)$.
We must produce some $\Gamma^*$ and $m^*$ such that
\begin{enumerate}
\item
$\Gamma^*$ is satisfiable if $\Gamma$ is,
\item
$\Gamma^*$ has a $\BCnnv$-refutation of size $m^*$ if $\Gamma$~has an
$\tER$-refutation of size $m$, and
\item
$m^*$ is bounded by a polynomial in~$m$.
\end{enumerate}
We use Lemma~\ref{lem:ER_to_BC}, letting $\Gamma^*$
be $\Gamma \cup X^m$ for $X^m$ in variables disjoint from $\Gamma$, and letting
$m^*$ be the bound on the size of the
$\BCnnv$-refutation of $\Gamma^*$.
\end{proof}

\section{Simulations}\label{sec:simulations}

\subsection{\texorpdfstring{$\DRATnnv$}{DRAT-} simulates \texorpdfstring{$\DPRnnv$}{DPR-}}

The following relations were known between $\DBCnnv$,
$\DRATnnv$ and~$\DPRnnv$.

\begin{thmC}[\cite{KRPH:erDRAT}]\label{thm:BCsimRATnoNew}
$\DBCnnv$ simulates $\DRATnnv$. (Hence they are equivalent).
\end{thmC}

\begin{thmC}[\cite{HeuleBiere:Variable}]\label{thm:DRATsimDPRonevar}
Suppose $\Gamma$ has a $\tDPR$ refutation $\Pi$.
Then it has a $\tDRAT$ refutation constructible in polynomial time
from $\Pi$, using at most one variable not appearing in $\Pi$.
\end{thmC}

We prove:

\begin{thm}\label{thm:DRATnnvDPRnnv}
$\DRATnnv$ simulates $\DPRnnv$.
\end{thm}

Hence the systems $\DBCnnv$, $\DRATnnv$, $\DSPRnnv$
and $\DPRnnv$ are all equivalent.
The theorem relies on the following main lemma used
in the proof of Theorem~\ref{thm:DRATsimDPRonevar}.
We include a proof for completeness.

\begin{lemC}[\cite{HeuleBiere:Variable}]\label{lem:DRATsimDPRonevar}
Suppose $C$ is $\tPR$ with respect to $\Gamma$.
Then there is a polynomial size  $\tDRAT$ derivation of $\Gamma \cup \{ C \}$
from $\Gamma$, using at most one variable not
appearing in $\Gamma$ or~$C$.
\end{lemC}

\begin{proof}
We have $\Gamma_{\rest\alpha} \vdash_1 \Gamma_{\rest\tau}$,
where $\alpha = \olnot C$ and $\tau \vDash C$. Let $x$ be a new
variable.
We describe the construction step-by-step.

\emph{Step 1.} For each $D \in \Gamma$ which is not satisfied by $\tau$,
derive $D_{\rest\tau} \lor \olnot x$ by $\tRAT$ on~$\olnot x$.
This is possible, as $x$ does not appear anywhere yet.

\emph{Step 2.}
Derive $C \lor x$ by $\tRAT$ on $x$. The only clauses in which
$\olnot x$ appears are those of the form~$D_{\rest\tau} \lor \olnot x$
introduced in step~1, and from Lemma~\ref{lem:RUP_restriction}
and the assumption that
$\Gamma_{\rest\alpha} \vdash_1 \Gamma_{\rest\tau}$
we have that~$\Gamma \vdash_1 D_{\rest\tau} \lor C$.

\emph{Step 3.}
For each $E \in \Gamma$  satisfied by $\tau$,
derive $E \lor x$ by a $\vdash_1$ step and delete~$E$.

\emph{Step 4.}
For each literal $p$ in $\tau$, derive $\olnot x \lor p$
by $\tRAT$ on $p$. To see that this satisfies the $\tRAT$ condition,
consider any clause $G = G' \dotlor \olnot p$
with which $\olnot x \lor p$ could be resolved. If $\tau \vDash G$, then
by steps~2 and~3 above, $G$ must also contain $x$,
so the resolvent $G' \cup \olnot{x}$ is a tautology. If $\tau \not \vDash G$,
then $G$ must be one of the clauses $D \in \Gamma$
or $D_{\rest\tau} \lor \olnot x$ from  step~1,
which means that we have already derived~$G_{\rest\tau} \lor \olnot x$,
which subsumes the resolvent $G' \lor \olnot x$.

\emph{Step 5.}
Consider each clause $E \lor x$ introduced in step~2 or~3.
In either case $\tau \vDash E$, so $E$ contains some literal $p$ in $\tau$.
Therefore we can derive $E$ by resolving $E \lor x$ with $\olnot x \lor p$.
Thus we derive $C$
and all clauses from $\Gamma$ deleted in step~3.

Finally delete all the new clauses except for $C$.
\end{proof}

\begin{defi}
Let $\Gamma$ be a set of clauses and $x$ any variable.
Then $\Gamma^{(x)}$  consists of
every clause in $\Gamma$ which does not mention $x$,
together with
every clause of the form $E \lor F$
where both $x \dotlor E$ and $\olnot x \dotlor F$ are in $\Gamma$.
\end{defi}

In other words, $\Gamma^{(x)}$ is formed from $\Gamma$ by
doing all possible resolutions with respect to~$x$ and then deleting all
clauses containing either $x$ or~$\olnot x$. (This is exactly like the
first step of the Davis-Putnam procedure. In~\cite{Kullmann:GeneralizationER} the
notation $\mathrm{DP}_x$ is used instead of $\Gamma^{(x)}$.)

\begin{lem}\label{lem:Gamma_from_(x)}
There is a polynomial size $\tDRAT$ derivation
of $\Gamma$ from $\Gamma^{(x)}$,
using only variables  from~$\Gamma$.
\end{lem}

\begin{proof}
We first derive every clause of the form $E \dotlor x$ in $\Gamma$,
by $\tRAT$ on $x$. As $\olnot{x}$ has not
appeared yet, the $\tRAT$ condition is satisfied.
Then we derive each clause of the form $F \dotlor \olnot{x}$ in $\Gamma$,
by $\tRAT$ on $\olnot{x}$. The only possible resolutions are with clauses of the
form $E \dotlor x$ which we have just introduced, but in this case
either $E\cup F$ is tautological or
$E \lor F$ is in $\Gamma^{(x)}$ so
 $\Gamma^{(x)} \vdash_1 \olnot{x} \lor F \lor E$.
Finally we delete all clauses not in $\Gamma$.
\end{proof}

The next two lemmas show that, under suitable conditions,
if we can derive $C$ from $\Gamma$ in $\DPRnnv$, then we can
derive it from $\Gamma^{(x)}$.
We will use a kind of normal form for $\tPR$ inferences. Say that a clause~$C$
is $\tPR_0$ with respect to $\Gamma$ if there is a partial assignment $\tau$
such that $\tau \vDash C$, all variables in~$C$ are in $\tdom(\tau)$, and
\begin{equation}\label{eq:PR0condition}
C \lor \Gamma_{\rest\tau} \subseteq \Gamma.
\end{equation}
(Recall that the notation $C \lor \Gamma_{\rest\tau}$ means
the set of clauses $C \lor D$ for $D \in \Gamma_{\rest\tau}$.)
The $\tPR_0$ inference rule lets us derive $\Gamma \cup \{ C \}$ from $\Gamma$
when (\ref{eq:PR0condition}) holds.
Letting $\alpha = \olnot{C}$
it is easy to see that (\ref{eq:PR0condition})
implies~$\Gamma_{\rest\tau} \subseteq \Gamma_{\rest\alpha}$,
so in particular
$\Gamma_{\rest\alpha} \vdash_1 \Gamma_{\rest\tau}$,
and hence this is a special case of the $\tPR$ rule.

\begin{lem}\label{lem:PR_to_PR_0}
Any $\tPR$ inference can be replaced with a $\tPR_0$ inference
together with polynomially many
$\vdash_1$ and deletion steps, using no new variables.
\end{lem}

\begin{proof}
Suppose $\Gamma_{\rest\alpha} \vdash_1 \Gamma_{\rest\tau}$,
where $\alpha = \olnot C$ and $\tau \vDash C$.
By Lemma~\ref{lem:PRnormal} we may assume that $\tdom(\alpha) \subseteq \tdom(\tau)$
so $\tdom(\tau)$ contains all variables in $C$.
Let $\Delta = C \lor \Gamma_{\rest\tau}$ and $\Gamma^* = \Gamma \cup \Delta$.
Note that~$\Delta_{\rest\tau}$ is empty, as $\tau$ satisfies~$C$.
This implies that
$C \lor \Gamma^*_{\rest\tau} = C \lor \Gamma_{\rest\tau} \subseteq \Gamma^*$,
so $C$ is $\tPR_0$ with respect to~$\Gamma^*$.
Furthermore the condition $\Gamma_{\rest\alpha} \vdash_1 \Gamma_{\rest\tau}$
and Lemma~\ref{lem:RUP_restriction}
imply that every clause in $\Delta$ is derivable from $\Gamma$
by a $\vdash_1$ step.
Thus we can derive $\Gamma^*$ from~$\Gamma$ by $\vdash_1$ steps,
then introduce $C$ by the $\tPR_0$ rule, and recover $\Gamma \cup \{ C \}$ by
deleting everything else.
\end{proof}

\begin{lem}\label{lem:PR_0_Gamma(x)}
Suppose $C$ is $\tPR_0$ with respect to $\Gamma$, witnessed by $\tau$
with~$x \notin \tdom(\tau)$. Then $C$ is $\tPR_0$ with
respect to $\Gamma^{(x)}$.
\end{lem}

\begin{proof}
The $\tPR_0$ condition implies that the variable $x$ does not occur in $C$.
We are given that $C \lor\penalty10000 \Gamma_{\rest\tau} \subseteq\penalty10000 \Gamma$
and want to show that
$C \lor \Gamma^{(x)}_{\rest\tau} \subseteq \Gamma^{(x)}$.
So let $D \in \Gamma^{(x)}$ with $\tau \not \vDash D$.
First suppose $D$ is in $\Gamma$ and  $x$ does not occur in~$D$.
Then $C \lor D_{\rest\tau} \in \Gamma$ by assumption,
so $C \lor D_{\rest\tau} \in \Gamma^{(x)}$.
Otherwise, $D=E \lor F$ where both $E \dotlor x$ and $F \dotlor \olnot x$
are in $\Gamma$. Then by assumption both $C \lor E_{\rest\tau} \lor x$
and $C \lor F_{\rest\tau} \lor \olnot x$ are in $\Gamma$.
Hence $C \lor D_{\rest\tau} = C \lor E_{\rest\tau} \lor F_{\rest\tau} \in \Gamma^{(x)}$.
\end{proof}

We can now prove Theorem~\ref{thm:DRATnnvDPRnnv},
that $\DRATnnv$ simulates $\DPRnnv$.

\begin{proof}[Proof of Theorem~\ref{thm:DRATnnvDPRnnv}]
We are given a $\DPRnnv$ refutation of some set  $\Delta$,
using only the variables in~$\Delta$.
By Lemma~\ref{lem:PR_to_PR_0} we may assume without loss
of generality that the refutation uses only $\vdash_1$, deletion and $\tPR_0$ steps.
Consider a $\tPR_0$ inference in this refutation, which derives
$\Gamma \cup \{ C \}$ from a set of clauses $\Gamma$,
witnessed by a partial assignment~$\tau$.
We want to derive $\Gamma \cup \{ C \}$ from $\Gamma$ in $\tDRAT$
using only variables in~$\Delta$.

Suppose $\tau$ is a total assignment to all variables in $\Gamma$.
The set $\Gamma$
is necessarily unsatisfiable,
as otherwise it could not occur as a line in a refutation.
Therefore $\Gamma_{\rest\tau}$ is simply~$\bot$,
so the $\tPR_0$ condition tells us that $C \in \Gamma$
and we do not need to do anything.

Otherwise, there is some variable $x$ which occurs in $\Gamma$
but is outside the domain of $\tau$, and thus in particular
does not occur in $C$.
We first use $\vdash_1$ and deletion steps to replace $\Gamma$ with
$\Gamma^{(x)}$.
By Lemma~\ref{lem:PR_0_Gamma(x)}, $C$ is $\tPR_0$, and thus $\tPR$, with respect to
$\Gamma^{(x)}$.
By Lemma~\ref{lem:DRATsimDPRonevar}  there
is a short $\tDRAT$ derivation
of $\Gamma^{(x)} \cup \{ C \}$ from $\Gamma^{(x)}$,
using one new variable which does not occur in $\Gamma^{(x)}$ or $C$.
We choose $x$ for  this variable.
Finally, observing that here
$\Gamma^{(x)} \cup \{ C \} = {(\Gamma \cup \{ C \})}^{(x)}$,
we  recover $\Gamma \cup \{C \}$ using Lemma~\ref{lem:Gamma_from_(x)}.
\end{proof}

\subsection{Towards a simulation of \texorpdfstring{$\PRnnv$}{PR-} by \texorpdfstring{$\SPRnnv$}{SPR-}}%
\label{sec:partialPRnnvSPRnnv}

Our next result shows how to replace a $\tPR$ inference
with $\tSPR$ inferences, without additional variables.
It is not a
polynomial simulation of $\PRnnv$ by $\SPRnnv$ however, as it
depends exponentially on the ``discrepancy'' as defined next.
Recall that $C$ is $\tPR$ with respect to $\Gamma$ if
$\Gamma_{\rest\alpha} \vdash_1 \Gamma_{\rest\tau}$, where
$\alpha = \olnot{C}$ and $\tau$ is a partial assignment
satisfying~$C$.
We will keep this notation throughout this section.
$C$ is $\tSPR$ with respect to $\Gamma$
if additionally $\tdom(\tau) = \tdom(\alpha)$.

\begin{defi}\label{def:discrepancy}
The \emph{discrepancy} of a $\tPR$ inference
 is  $|\tdom(\tau)\setminus\tdom(\alpha)|$.
 That is, it is the number of variables which are assigned by~$\tau$
but not by~$\alpha$.
\end{defi}

\begin{thm}\label{thm:SPRsimPRnoNew_disc}
Suppose that $\Gamma$ has a $\tPR$ refutation $\Pi$ of
size~$S$ in which every $\tPR$ inference has
discrepancy bounded by $\delta$. Then $\Gamma$ has a
$\tSPR$ refutation of size $O(2^\delta S)$ which does not
use any variables not present in $\Pi$.
\end{thm}

When the discrepancy is logarithmically bounded,
Theorem~\ref{thm:SPRsimPRnoNew_disc} gives
polynomial size $\tSPR$ refutations automatically.
We need a couple of lemmas before proving
the theorem. 

\begin{lem}\label{lem:AlphaPlus}
Suppose $\Gamma_{\rest\alpha} \vdash_1 \Gamma_{\rest \tau}$
and $\beta$ is a partial assignment extending $\alpha$,
such that $\tdom(\beta) \subseteq \tdom(\tau)$.
Then $\Gamma_{\rest\beta} \vdash_1 \Gamma_{\rest \tau}$
\end{lem}

\begin{proof}
Suppose $E \in \Gamma_{\rest\tau}$. Then $E$ contains no
variables from $\beta$, so $\olnot{E}_{\rest\beta} = \olnot{E}$,
and by assumption there is a
refutation $\Gamma_{\rest\alpha}, \olnot{E} \vdash_1 \bot$.
Thus $\Gamma_{\rest\beta}, \olnot{E} \vdash_1 \bot$ by
Fact~\ref{fac:unit_propagation}.
\end{proof}

\begin{proof}[Proof of Theorem~\ref{thm:SPRsimPRnoNew_disc}.]
Our main task is to show that a $\tPR$ inference with discrepancy
bounded by~$\delta$
can be simulated by multiple $\tSPR$ inferences, while bounding
the increase in proof size in terms of~$\delta$.
Suppose $C$ is derivable from~$\Gamma$ by a $\tPR$ inference.
That is, $\Gamma_{\rest \alpha} \vdash_1 \Gamma_{\rest\tau}$
where $\alpha = \olnot C$ and $\tau \vDash C$, and
by Lemma~\ref{lem:PRnormal} we may assume
that $\tdom(\tau)\supseteq \tdom(\alpha)$.
List the variables in $\tdom(\tau) \setminus \tdom(\alpha)$
as~$p_1, \ldots, p_s$, where $s \le \delta$.

Enumerate as $D_1, \dots, D_{2^s}$ all clauses
containing exactly the variables $p_1, \ldots, p_s$
with some pattern of negations.
Let $\sigma_i = \olnot { C \lor D_i}$,
so that $\sigma_i \supseteq \alpha$ and
$\tdom(\sigma_i) = \tdom(\tau)$.
By Lemma~\ref{lem:AlphaPlus},
$\Gamma_{\rest\sigma_i} \vdash_1 \Gamma_{\rest\tau}$.
Since $\tau \vDash C \lor D_j$ for every $j$,
in fact
\[ \Gamma_{\rest\sigma_i} \vdash_1
           {(\Gamma \cup \{C \lor\penalty10000 D_1,
              \dots, C\lor\penalty10000 D_{i-1}\})}_{\rest\tau}.
\]
Thus we may introduce all clauses $C \lor D_1,
\dots, C \lor D_{2^s}$ one after another by $\tSPR$ inferences.
We can then use $2^s-1$ resolution steps to derive~$C$.

The result is a set~$\Gamma^\prime \supseteq \Gamma$
which contains $C$ plus many extra clauses  subsumed by~$C$,
which must be carried through the rest of the refutation,
as we do not have the deletion rule.
But by Lemma~\ref{lem:subsumeNoDelete}(a) this is not a problem, 
as the presence of these
additional subsumed clauses does not affect the validity
of later $\tPR$ inferences.
\end{proof}

\section{Upper bounds for some hard tautologies}
\label{sec:Upperbounds}

This section proves that $\SPRnnv$
 --- without new variables --- can
give polynomial size refutations for many of the
usual ``hard'' propositional principles.
Heule, Kiesl and Biere~\cite{HKB:StrongExtensionFree,HKB:NoNewVariables}
showed that the tautologies based on the
pigeonhole principle (PHP)
have polynomial size $\SPRnnv$ proofs,
and Heule and Biere discuss polynomial size
$\PRnnv$ proofs of the Tseitin tautologies and the 2--1 pigeonhole principle
in~\cite{HeuleBiere:Variable}.
The $\SPRnnv$ proof of the PHP tautologies can be viewed as a
version of the original extended resolution proof of PHP given
by Cook and Reckhow~\cite{CookReckhow:proofs};
see also~\cite{Kullmann:GeneralizationER} for an
adaptation of the original proof to use $\tBC$ inferences.

 Here we describe polynomial size $\SPRnnv$ proofs
 for several well-known principles,
namely the pigeonhole principle, the bit pigeonhole principle,
the parity principle, the clique-coloring principles,
and the Tseitin tautologies.
We also
show that orification, xorification, and typical cases of lifting can be handled in $\SPRnnv$.

The existence of such small proofs is surprising,
 since they use only clauses
in the original literals, and it is well-known that such clauses are limited
in what they can express. However, $\SPRnnv$ proofs can
exploit the underlying symmetries of the principles to introduce
new clauses, in effect arguing that properties can
be assumed to hold ``without loss of generality'' (see~\cite{RebolaPardoSuda:SatPreserving}).

It is open whether extended resolution,
or the Frege proof system, can be simulated
by $\PRnnv$ or~$\DPRnnv$, or more generally by~$\DSRnnv$.
The examples below show that
any separation of these systems must involve a new technique.

Our proofs use the same basic idea as the sketch in
Example~\ref{ex:PHP_in_SR}.
One complication is that we are now
 working with $\tSPR$ rather than $\tSR$ inferences.
This requires us to make the individual inferences more complicated
-- for example the assignments $\alpha$ in
Example~\ref{ex:PHP_in_SR} set one pigeon, while
those in Section~\ref{sec:PHP} below set two pigeons.
Another is that we want to avoid using any
deletion steps. This means that, when showing
that an $\tSPR$ inference is valid,
 we have to
consider every clause introduced so far.
For this reason we will do all necessary $\tSPR$ inferences
at the start, in a careful order,
before we do any resolution steps.
This is the purpose of Lemma~\ref{lem:SPR_symmetry_2}.

\begin{defi}
A \emph{$\Gamma$-symmetry} is an invertible substitution $\pi$ such that
$\Gamma_{\rest\pi}  = \Gamma$.
\end{defi}

We will use the observation that, if $\pi$ is
a $\Gamma$-symmetry and $\alpha = \olnot{C}$ is a partial assignment, then by
Lemma~\ref{lem:substitution_composition} we have
\[
\Gamma_{\rest\alpha} = {(\Gamma_{\rest\pi})}_{\rest\alpha}
=\Gamma_{\rest \alpha \circ \pi}.
\]
Hence, if $\alpha\circ\pi \vDash C$, we can infer $C$
from $\Gamma$ by an $\tSR$ inference with $\tau = \alpha\circ\pi$.
If furthermore all literals in the domain and image of $\pi$ are in $\tdom(\alpha)$, then
$\alpha \circ \pi$ behaves as a partial assignment
and $\tdom(\alpha \circ \pi) = \tdom(\alpha)$,
so this becomes
an  $\tSPR$ inference.

We introduce one new piece of notation, writing $\olnot{\alpha}$
for the clause expressing that the partial assignment $\alpha$ does not hold
(so $C = \olnot{\alpha}$ if and only if $\alpha = \olnot{C}$).
Two partial assignments are called \emph{disjoint}
if their domains are disjoint.
The next lemma describes sufficient conditions for introducing,
successively, the clauses $\olnot{\alpha_i}$ for $i=0,1,2,\ldots$
using only $\tSPR$ inferences.

\begin{lem}\label{lem:SPR_symmetry_2}
Suppose $(\alpha_0, \tau_0), \ldots, (\alpha_m, \tau_m)$
is a sequence of pairs of partial assignments such that for each $i$,
\begin{enumerate} 
\item 
$\Gamma_{\rest\alpha_i} = \Gamma_{\rest\tau_i}$
\item 
$\alpha_i$ and $\tau_i$ are contradictory and have the same domain
\item 
for all $j < i$, the assignments $\alpha_j$ and $\tau_i$ are either disjoint or contradictory.
\end{enumerate}
Then we can derive $\Gamma \cup \{ \olnot{\alpha_i} : i=0,\dots,m\}$
from $\Gamma$ by a sequence of $\tSPR$ inferences.
\end{lem}

\begin{proof}
We write $C_i$ for $\olnot{\alpha_i}$.
By item 2,~$\tau_i \vDash C_i$.
Thus it is enough to show that for each $i$,
\[
{\big( \Gamma \cup \{ {C_0}, \ldots, {C_{i-1}} \} \big)}_{\rest\alpha_i}
\supseteq
{\big( \Gamma \cup \{ {C_0}, \ldots, {C_{i-1}} \} \big)}_{\rest\tau_i}.
\]
We have $\Gamma_{\rest\alpha_i} = \Gamma_{\rest\tau_i}$.
For $j<i$, either $\alpha_j$ and $\tau_i$ are disjoint
and consequently ${(C_j)}_{\rest\alpha_i} = {(C_j)}_{\rest\tau_i}=C_j$,
or they are contradictory and so $\tau_i \vDash C_j$ and $C_j$ vanishes
from the right hand side.
\end{proof}

In the lemma, if we added to (2) the condition that $\alpha_i$
and $\tau_i$ disagree on only a single variable, then
by Theorem~\ref{thm:RATandLPR} we could derive
the clauses $\olnot{\alpha_i}$ by $\tRAT$ inferences rather than
needing $\tSPR$ inferences. However in the applications below
they typically differ on more than one variable,
so our proofs are in $\SPRnnv$, not in~$\RATnnv$.

\subsection{Pigeonhole principle}\label{sec:PHP}

Let $n\ge 1$ and $[n]$ denote $\{0,\ldots,n{-}1\}$.
The pigeonhole principle $\PHP_n$ consists of the clauses
\begin{align*}
\bigor_{j\in [n]} p_{i,j}
     & \quad \text{for each fixed $i\in[n+1]$}
     &&\hbox{(pigeon axioms)} \\
\olnot {p_{i,j}} \lor \olnot{p_{i^\prime,j}}
     & \quad \text{for all $i < i^\prime\in[n+1]$ and $j \in [n]$}
     &&\hbox{(hole axioms)}.
\end{align*}

\begin{thmC}[\cite{HKB:StrongExtensionFree}]\label{thm:PHP_SPR_2}
$\PHP_n$ has polynomial size $\SPRnnv$ refutations.
\end{thmC}

\begin{proof}
Our strategy is to first derive all unit clauses $\olnot{p_{j,0}}$
for $j>0$, which effectively takes pigeon~$0$ and hole~$0$ out of the picture
and reduces $\PHP_n$ to a renamed instance of $\PHP_{n-1}$.
We repeat this construction to reduce to a renamed instance of $\PHP_{n-2}$, etc.
At each step, we will need to use several clauses introduced
by $\tSPR$ inferences. We  use Lemma~\ref{lem:SPR_symmetry_2}
to introduce all necessary clauses at one go at the start of the construction.

Let~$\alpha_{i,j,k}$ be the assignment setting
$p_{i,k}=1$,  $p_{j,i}=1$ and all other variables
$p_{\ell,k}, p_{\ell,i}$ for holes~$k$ and~$i$ to~$0$.
Let~$\pi_{k,i}$ be the $\PHP_n$-symmetry which switches holes~$k$ and $i$,
that is,  maps $p_{\ell, i} \mapsto p_{\ell, k}$ and $p_{\ell, k} \mapsto p_{\ell,i}$
for every pigeon~$\ell$. Let~$\tau_{i,j,k}$ be $\alpha_{i,j,k} \circ \pi_{k,i}$, so in particular~$\tau_{i,j,k}$
sets $p_{i,i}=1$ and $p_{j,k}=1$.
By the properties of symmetries,
we have ${(\PHP_n)}_{\rest \alpha_{i,j,k}} = {(\PHP_n)}_{\rest\tau_{i,j,k}}$.

For $i=0, \dots, n-2$ define
\[
A_i := \{ (\alpha_{i,j,k}, \tau_{i,j,k}) : i< j < n+1, \ i<k <n\}.
\]
Any $\tau_{i,j,k}$ appearing in $A_i$ contradicts every
 $\alpha_{i,j',k'}$ appearing in $A_i$, since they disagree about
 which pigeon maps to hole $i$. On the other
hand, if $i'<i$ and $\alpha_{i',j',k'}$ appears in $A_{i'}$
and is not disjoint from~$\tau_{i,j,k}$, then they must share some hole. So either $i=k'$ or $k=k'$, and in either case they disagree
about hole $k'$.

Hence we can apply Lemma~\ref{lem:SPR_symmetry_2}
to derive all clauses $\olnot{\alpha_{i,j,k}}$ such that $i < j < n+1$
and $i < k < n$.
Note $\olnot{\alpha_{i,j,k}}$ is the clause
$\olnot{p_{i,k}} \lor \bigvee_{\ell\not=i}p_{\ell,k} \lor \olnot{p_{j,i}}
   \lor \bigvee_{\ell\not=j}p_{\ell,i}$, which
we resolve
with hole axioms to get~$\olnot{p_{i,k}} \lor \olnot{p_{j,i}}$.

Now we  use induction on $i=0, \dots, n-1$ to derive all unit clauses
$\olnot{p_{j,i}}$ for all $j$ with $i<j<n+1$.
Fix $j>i$. For each hole $k>i$ we have
$\olnot{p_{i,k}} \lor \olnot{p_{j,i}}$ (or  if $i=n-1$ there is no such~$k$).
We have $\olnot{p_{i,i}} \lor \olnot{p_{j,i}}$ since it is a hole axiom,
and for each $k<i$, we have $\olnot{p_{i,k}}$ from the inductive hypothesis.
Resolving all these with the  axiom~$\bigvee_k p_{i,k}$
gives $\olnot {p_{j,i}}$.

Finally  the unit clauses $\olnot{p_{n,i}}$ for $i<n$ together
contradict the  axiom~$\bigvee_i p_{n,i}$.
\end{proof}

\subsection{Bit pigeonhole principle}\label{sec:BPHP}

Let $n=2^k$.
The \emph{bit pigeonhole principle} contradiction, $\BPHP_n$,
asserts that each of $n+1$ pigeons can be assigned a distinct $k$-bit binary string.
For each pigeon~$x$, with
 $0\le x<n+1$, it has variables $p^x_1, \dots, p^x_k$ for
the bits of the string assigned to~$x$. We think of strings $y \in {\{0,1\}}^k$ as holes.
When convenient we will identify holes with numbers $y<n$.
We write $(x \pigeonto y)$ for the conjunction $\bigwedge_i (p^x_i = y_i)$
asserting that pigeon~$x$ goes to hole~$y$,
where $p^x_i = 1$ is the literal $p^x_i$
and $p^x_i = 0$ is the literal $\olnot{p^x_i}$, and where
$y_i$ is the $i$-th bit of~$y$.
We write $(x \notpigeonto y)$ for its negation: $\bigvee_i (p^x_i \neq y_i)$.
The axioms of $\BPHP_n$ are then
\[
(x \notpigeonto y) \lor (x' \notpigeonto y)
\quad
\hbox{for all holes $y$ and all distinct pigeons~$x, x'$.}
\]
Notice that the set $\{ (x \notpigeonto y) : y<n \}$ consists of the $2^k$ clauses
containing the variables $p^x_1, \dots, p^x_k$ with all
patterns of negations. We can derive $\bot$ from this set in
$2^k{-}1$ resolution steps.

\begin{thm}\label{thm:SPR_BPHP}
The $\BPHP_n$ clauses
have polynomial size $\SPRnnv$ refutations.
\end{thm}

The theorem is proved below. It is essentially the same
as the proof of PHP in~\cite{HKB:StrongExtensionFree}
(or Theorem~\ref{thm:PHP_SPR_2} above).
For each  $m < n-1$ and each pair $x,y>m$,
we define a clause
\[
C_{m,x,y} := (m \notpigeonto y) \lor (x \notpigeonto m).
\]
Note we allow $x=y$.
Let $\Gamma$ be the set of all such clauses $C_{m,x,y}$.
We will show these clauses can be introduced by $\tSPR$ inferences,
but first we show they suffice to derive $\BPHP_n$.

\begin{lem}\label{lem:BPHP_final_refutation}
$\BPHP_n \cup \Gamma$
 has a polynomial size resolution refutation.
\end{lem}

\begin{proof}
Using induction on $m = 0, 1,2, \ldots, n{-}1$ we derive all clauses
$(x \notpigeonto m)$ such that $x > m$.
So suppose $m < n$ and $x >m$.
For each $y>m$, we have the clause
$
(m \notpigeonto y) \lor (x \notpigeonto m),
$
as this is  $C_{m,x,y}$. We also have the clause
$
(m \notpigeonto m) \lor (x \notpigeonto m),
$
as this is an axiom of $\BPHP_n$. Finally, for each $m'<m$, we have
$
(m \notpigeonto m')
$
by the inductive hypothesis (or, in the base case $m=0$, there are no such clauses).
Resolving all these together gives $(x \notpigeonto m)$.

At the end we have in particular derived all the clauses
$(n \notpigeonto m)$ such that $m < n$.
Resolving all these clauses together yields $\bot$.
\end{proof}

Thus it is enough to show that we can introduce all clauses in $\Gamma$ using $\tSPR$ inferences.
We use Lemma~\ref{lem:SPR_symmetry_2}. For $m < n-1$ and each pair $x,y>m$,
define partial assignments
\begin{align*}
\alpha_{m,x,y} &:= (m \pigeonto y) \wedge (x \pigeonto m)\\
\tau_{m,x,y} &:= (m \pigeonto m) \wedge (x \pigeonto y)
\end{align*}
so that
$C_{m,x,y} = \olnot{\alpha_{m,x,y}}$ and
$\tau_{m,x,y} = \alpha_{m,x,y} \circ \pi$ where $\pi$ swaps all variables
for pigeons~$m$ and~$x$. Hence
${(\BPHP_n)}_{\rest\alpha_{m,x,y}}={(\BPHP_n)}_{\rest\tau_{m,x,y}}$
as required.

For the other conditions for Lemma~\ref{lem:SPR_symmetry_2},
first observe that assignments $\alpha_{m,x,y}$
and $\tau_{m,x',y'}$ are always inconsistent, since they map $m$ to different places.
Now suppose that $m<m'$ and $\alpha_{m,x,y}$ and~$\tau_{m',x',y'}$
are not disjoint. Then they must have some pigeon in common,
so either $m'=x$ or $x'=x$. In both cases $\tau_{m',x',y'}$
contradicts $(x \pigeonto m)$, in the first case because
it maps $x$ to $m'$,
and in the second  because it maps~$x$ to $y'$ with $y'>m'$.

\IGNORE{
The other conditions for Lemma~\ref{lem:SPR_symmetry_2}
follow from the following two observations:
\textbf{[To be rewritten, to match the lemma better]}
\begin{enumerate}[(i)]
\item
for any $x',y'$  the partial assignment $\tau[m,x',y']$ satisfies $C_{m,x,y}$
\item
if $m'>m$, then for any $x',y'$ either
\begin{enumerate}
\item
$\alpha[m',x',y']$ and $\tau[m',x',y']$ set no variables in $C_{m,x,y}$, or
\item
$\tau[m',x',y']$ satisfies $C_{m,x,y}$.
\end{enumerate}
\end{enumerate}
For (i), it is enough to observe that
$\tau[m,x',y']$ maps $m$ to $m$,
while $C_{m,x,y}$ contains $(m \notpigeonto y)$ with $y>m$.
For (ii), notice that $\alpha[m',x',y']$ and $\tau[m',x',y']$
both set precisely the variables for pigeons $m'$ and $x'$.
If this includes some variable
in $C_{m,x,y}$
then there is some pigeon in common, so either
$m'=x$ or~$x'=x$.
In both cases $\tau[m',x',y]$ satisfies $(x \notpigeonto m)$,
in the first case because it maps $x$ to $m'$,
and in the second  because it maps $x$ to $y'$ with $y'>m'$.}

\subsection{Parity principle}\label{sec:PRparity}

The \emph{parity principle} states that there is no (undirected) graph on
an odd number of vertices in which each vertex has degree exactly one
(see~\cite{Ajtai:ParityandPHP,BIKPP:nullstellensatz}).
For $n$ odd, let $\PAR_n$ be a set of clauses expressing (a violation of) the
parity principle on $n$ vertices, with variables $x_{i,j}$ for the $\binom n 2$
many values $0\le i < j < n$, where we identify the variable $x_{i,j}$ with~$x_{j,i}$.
We write $[n]$ for~$\{0,\ldots, n{-}1\}$.
$\PAR_n$ consists of the clauses
\begin{align*}
\bigor_{j\not= i}  x_{i,j} & \quad \hbox{for each fixed $i\in[n]$} &&\hbox{(``pigeon'' axioms)} \\
\olnot {x_{i,j}} \lor \olnot{x_{i,j^\prime}} & \quad \text{for all distinct $i,j,j^\prime \in [n]$} &&\hbox{(``hole'' axioms)}.
\end{align*}

\begin{thm}\label{thm:PRparity}
The $\PAR_n$ clauses
have polynomial size $\SPRnnv$ refutations.
\end{thm}

\begin{proof}
 Let $n=2m+1$.
For $i<m$ and distinct $j,k$ with $2i+1<j,k<n$
define $\alpha_{i,j,k}$ to be the partial assignment
which matches $2i$ to $j$ and $2i+1$ to $k$, and sets all other adjacent variables to $0$.
That is, $x_{2i,j}=1$ and $x_{2i, j'}=0$ for all $j' \neq j$,
and $x_{2i+1,k}=1$
and $x_{2i+1, k'}=0$ for all $k' \neq k$.
Similarly define $\tau_{i,j,k}$ to be the partial assignment
which matches $2i$ to $2i+1$ and $j$ to $k$, and sets all other adjacent variables to $0$,
so that $\tau_{i,j,k} =\alpha_{i,j,k} \circ \pi$ where $\pi$ swaps vertices
$2i+1$ and $j$.
It is easy to see that the conditions of Lemma~\ref{lem:SPR_symmetry_2} are satisfied.
Therefore, we can introduce all clauses  $\olnot{\alpha_{i,j,k}}$
by $\tSPR$ inferences.

We now inductively derive the unit clauses $x_{2i,2i+1}$ for $i=0, 1, \dots, m-1$.
Once we have these, refuting $\PAR_n$ becomes trivial.
So suppose we have $x_{2i',2i'+1}$ for all $i'<i$ and want to derive $x_{2i, 2i+1}$.
Consider any $r<2i$. First suppose $r$ is even, so $r=2m$ for some $m<i$.
We resolve the ``hole'' axiom $\olnot{x_{2m, 2i}} \lor \olnot{x_{2m, 2m+1}}$
with $x_{2m, 2m+1}$ to get $\olnot{x_{2m, 2i}}$, which is the same clause
as $\olnot{x_{2i, r}}$. A similar argument works for $r$ odd, and
we can also obtain $\olnot{x_{2i+1, r}}$ in a similar way.

Resolving the clauses $\olnot{x_{2i, r}}$ and $\olnot{x_{2i+1, r}}$ for $r<2i$
with the ``pigeon'' axioms for vertices $2i$ and $2i+\penalty10000 1$ gives clauses
\[
x_{2i,2i+1} \lor \bigvee_{r>2i+1} x_{2i,r}
\quad \hbox{and} \quad
x_{2i, 2i+1} \lor \bigvee_{r >2i+1} x_{2i+1,r}.
\]
Now by resolving clauses $\olnot{\alpha_{i,j,k}}$ with suitable ``hole'' axioms
we can get $\olnot{x_{2i,j}} \lor \olnot{x_{2i+1, k}}$
for all distinct~$j,k>2i+1$. Resolving these with the clauses above gives
$x_{2i,2i+1}$, as required.
\end{proof}

\subsection{Clique-coloring principle}\label{sec:PRcliqueColoring}

The \emph{clique-coloring principle} $\CC_{n,m}$ states, informally, that a graph
with $n$ vertices cannot have
both a clique of size~$m$ and a coloring of size $m-1$
(see~\cite{Krajicek:interpolation,Pudlak:monotone}). For~$m \le n$
integers, $\CC_{n,m}$ uses  variables
$p_{a,i}$, $q_{i,c}$ and $x_{i,j}$ where $a \in [m]$
and $c\in[m{-}1]$ and~$i, j \in [n]$ with~$i\not= j$. Again, $x_{i,j}$
is identified with $x_{j,i}$.  The intuition is that $x_{i,j}$ indicates
that vertices $i$ and~$j$ are joined by an edge, $p_{a,i}$ asserts
that $i$ is the $a$-th vertex of a clique, and $q_{i,c}$ indicates that
vertex~$i$ is assigned color~$c$.
We list the clauses of $\CC_{n,m}$ as
\begin{enumerate}[(i)]
\item
$\bigor\nolimits_{i}  p_{a,i}$ \
for each  $a\in[m]$
\item
$\olnot {p_{a,i}} \lor \olnot{p_{a',i}}$ \
for  distinct $a,a^\prime \in[m]$ and each $i \in [n]$
\item
$\bigor\nolimits_{c} q_{i,c}$ \
for each  $i\in[n]$
\item
$\olnot {q_{i,c}} \lor \olnot{q_{i,c^\prime}}$ \
for each $i\in[n]$ and distinct $c,c^\prime \in [m{-}1]$
\item
$\olnot{p_{a,i}} \lor \olnot{p_{a^\prime, j}} \lor x_{i,j}$ \
for each distinct $a,a^\prime\in[m]$ and
distinct $i,j\in[n]$
\item
$\olnot {q_{i,c}} \lor \olnot{q_{j,c}} \lor \olnot{x_{i,j}}$ \
for each $c\in[m{-}1]$ and distinct $i,j \in [n]$.
\end{enumerate}

\begin{thm}\label{thm:PRcliqueColoring}
The $\CC_{n,m}$ clauses
have polynomial size $\SPRnnv$ refutations.
\end{thm}

\begin{proof}
The intuition for the $\SPRnnv$ proof is that we introduce clauses
stating that the first $r$ clique members are assigned vertices that
are colored by the first $r$ colors; iteratively for $r=1,2,\ldots$.

Write $(a \pigeonto i \pigeonto c)$
for the assignment which sets
\begin{align*}
p_{a,i} &= 1 \text{\quad and \quad $p_{a',i}=0$ for all $a' \neq a$}\\
q_{i,c} &= 1 \text{\quad and \quad $q_{i,c'}=0$ for all $c' \neq c$.}
\end{align*}
For all $r<m-2$, all indices $a>r$, all colors $c>r$ and all distinct vertices $i,j \in [n]$, define
\begin{align*}
\alpha^r_{a,i,j,c} &:=
 (a \pigeonto j \pigeonto r) \land (r \pigeonto i \pigeonto c)\\
\tau^r_{a,i,j,c} &:=
 (a \pigeonto j \pigeonto c) \land (r \pigeonto i \pigeonto r).
\end{align*}

Let $\Gamma$ consist of axioms (i), (ii) and (v), containing $p$ and~$x$ variables
but no~$q$ variables,
and let~$\Delta$ consist of the remaining axioms (iii), (iv) and (vi),
containing $q$ and~$x$ variables but no $p$ variables.
Let us write $\alpha$ for $\alpha^r_{a,i,j,c}$ and $\tau$
for $\tau^r_{a,i,j,c}$.
Then $\Gamma_{\rest \alpha} = \Gamma_{\rest \tau}$ since
$\alpha$ and $\tau$ are the same on $p$ variables. Let $\alpha'$ and $\tau'$
be respectively $\alpha$ and $\tau$ restricted to $q$ variables.
Then $\Delta_{\rest \alpha'} = \Delta_{\rest \tau'}$
since $\tau' = \alpha' \circ \pi$ where $\pi$
is the $\Delta$-symmetry which swaps vertices $i$ and $j$.
Hence also $\Delta_{\rest \alpha} = \Delta_{\rest \tau}$.

We will show that the conditions of Lemma~\ref{lem:SPR_symmetry_2} are satisfied,
so we can introduce all clauses~$\olnot{\alpha^r_{a,i,j,c}}$
by $\tSPR$ inferences. The first condition was just discussed.
For the second condition, first notice that
$\alpha^r_{a,i,j,c}$ and $\tau^r_{a,i,j,c}$ set the same variables.

Now suppose $r,a,c,i,j$ are such that $r<a<m$, that $r<c<m{-}1$, and that $i,j\in [n]$ are distinct.
Suppose  $r',a',c',i',j'$ satisfy the same conditions, with $r' \le r$.
We want to show that if $\tau := \tau^r_{a,i,j,c}$ and $\alpha := \alpha^{r'}_{a',i',j',c'}$
are not disjoint, then they are contradictory. Notice that showing this
will necessarily use the literals $\olnot{p_{a',i}}$ and~$\olnot{q_{i,c'}}$
in the definition of our assignments, and that it will be enough to show
that~$\alpha$ and~$\tau$ disagree about either which index or which color is assigned to a vertex~$i$.
First suppose~$r'=r$. Assuming $\alpha$ and $\tau$ are not disjoint, we must be in
one of the following four cases.
\begin{enumerate}
\item
$i'=i$. Then $\alpha$ maps vertex $i$ to color $c'>r$ while $\tau$ maps $i$ to color $r$.
\item
$i'=j$. Then $\alpha$ maps index $r<a$ to vertex $j$ while $\tau$ maps index $a$ to $j$.
\item
$j'=i$. Then $\alpha$ maps index $a'>r$ to vertex $i$ while $\tau$ maps index $r$ to $i$.
\item
$j'=j$. Then $\alpha$ maps vertex $j$ to color $r<c$ while $\tau$ maps $j$ to color $c$.
\end{enumerate}
Now suppose $r'<r$. Assuming $\alpha$ and $\tau$ are not disjoint, we have the same cases.
\begin{enumerate}
\item
$i'=i$. Then $\alpha$ maps index $r'<r$ to vertex $i$ while $\tau$ maps index $r$ to $i$.
\item
$i'=j$. Then $\alpha$ maps index $r'<a$ to vertex $j$ while $\tau$ maps index $a$ to $j$.
\item
$j'=i$. Then $\alpha$ maps vertex $i$ to color $r'<r$ while $\tau$ maps $i$ to color $r$.
\item
$j'=j$. Then $\alpha$ maps vertex $j$ to color $r'<c$ while $\tau$ maps $j$ to color $c$.
\end{enumerate}

\noindent
Thus the conditions are met and we can introduce the clauses~$\olnot{\alpha^r_{a,i,j,c}}$, that is,
\[
\olnot{p_{a,j}} \vee \bigvee_{a'\neq a} p_{a',j}
\vee
\olnot{q_{j,r}} \vee \bigvee_{r'\neq r} q_{j,r'}
\vee
\olnot{p_{r,i}} \vee \bigvee_{r'\neq r} p_{r',i}
\vee
\olnot{q_{i,c}} \vee \bigvee_{c'\neq c} q_{i,c'},
\]
for all $r<a<m$, all $r<c<m{-}1$ and all
distinct $i,j \in [n]$.
Now let~$C^r_{a,i,j,c}$ be the clause
\[
\olnot{p_{a,j}} \lor \olnot{q_{j,r}}
\ \lor \olnot{p_{r,i}} \lor \olnot{q_{i,c}}.
\]
We derive this by resolving $\olnot{\alpha^r_{a,i,j,c}}$
with instances of axiom~(ii) to remove the literals $p_{a',j}$ and $p_{r',i}$
and then with
instances of axiom~(iv) to remove the literals $q_{j,r'}$ and~$q_{i,c'}$.
We now want to derive, for each $r$, each $a$ with $r<a<m$ and each
$j \in [n]$, the clause
\begin{equation}\label{eq:clause}
\olnot{p_{a,j}} \lor \bigvee_{c>r} q_{j,c}
\end{equation}
which can be read as ``if $a>r$ goes to $j$, then $j$ goes
to some $c>r$''. Intuitively, this removes indices
and colors $0, \dots, r$ from $\CC_{n,m}$, thus reducing it to
a CNF isomorphic
to~$\CC_{n,m-r-1}$.

Suppose inductively that we have
already derived (\ref{eq:clause}) for all $r'<r$.  In particular
 we have derived~$\olnot{p_{a,j}} \lor \bigvee_{c>r-1} q_{j,c}$,
or for $r=0$ we use the axiom $\bigvee_{c} q_{j,c}$.
We resolve this with the clauses~$C^r_{a,i,j,c}$ for all $c>r$
to get
\begin{equation} \label{eq:pqpq}
\olnot{p_{a,j}} \lor \olnot{q_{j,r}}
\ \lor \olnot{p_{r,i}} \lor {q_{i,r}}.
\end{equation}
By resolving together suitable instances of axioms~(v)
and~(vi) we obtain
\[
\olnot{p_{a,j}} \lor \olnot{p_{r,i}}
 \lor \olnot{q_{j,r}} \lor \olnot{q_{i,r}}
\]
and resolving this with (\ref{eq:pqpq})
removes the variable $q_{i,r}$ to give
$\olnot{p_{a,j}} \lor \olnot{q_{j,r}}
\ \lor \olnot{p_{r,i}}$.
We derive this for every $i$, and
then resolve with the axiom $\bigvee_i p_{r,i}$
to get $\olnot{p_{a,j}} \lor \olnot{q_{j,r}}$, and finally
again with our inductively given clause
$\olnot{p_{a,j}} \lor \bigvee_{c>r-1} q_{j,c}$
to get $\olnot{p_{a,j}} \lor \bigvee_{c>r} q_{j,c}$
as required.
\end{proof}



\subsection{Tseitin tautologies}\label{sec:PRtseitinAgain}

The \emph{Tseitin tautologies} $\TS_{G,\gamma}$ are well-studied hard examples
for many proof systems (see~\cite{Tseitin:derivation,Urquhart:hardresolution}).
Let $G$ be an undirected graph with $n$ vertices, with each vertex~$i$ labelled with a
charge~$\gamma(i) \in \{0,1\}$ such that the total charge on $G$ is odd.
For each edge~$e$ of~$G$ there is a variable~$x_e$.
Then $\TS_{G,\gamma}$ consists of clauses
expressing that, for each vertex~$i$, the
parity of the values~$x_e$ over the edges $e$ touching~$i$ is equal to the charge~$\gamma(i)$.
For a vertex $i$ of degree $d$, this requires~$2^{d-1}$ clauses,
using one clause to
rule out each assignment to the edges touching $i$ with the wrong parity.
 If $G$ has constant
degree then this has size polynomial in $n$, but in general the size
may be exponential in $n$.
It is well-known to be unsatisfiable.

The next lemma is a basic property of Tseitin contradictions.
Note that it does not depend on $\gamma$. By cycle
we mean a simple cycle, with no repeated vertices.

\begin{lem}\label{lem:Tseitin_symmetry}
Let $K$ be any cycle in $G$. Then the substitution $\pi_K$ which flips
the sign of every literal on $K$ is a $\TS_{G, \gamma}$-symmetry.
\end{lem}

\begin{lem}\label{lem:short_cycle}
If every node in $G$ has degree at least $3$, then $G$ contains
a cycle of length at most~$2 \log n$.
\end{lem}

\begin{proof}
Pick any vertex $i$ and let $H$ be the subgraph consisting of all
vertices reachable from~$i$ in at most $\log n$ steps.
Then $H$ cannot be a tree, as otherwise by the assumption on degree it
would contain more than $n$ vertices. Hence it must contain
some vertex reachable from $i$ in two different ways.
\end{proof}

\begin{thm}\label{thm:TS_SPRproofs}
The $\TS_{G, \gamma}$ clauses have polynomial size $\SPRnnv$ refutations.
\end{thm}

\begin{proof}
We will construct a sequence of triples
$(G_0, \gamma_0, \ell_0), \dots, (G_m, \gamma_m, \ell_m)$
where $(G_0, \gamma_0)$ is $(G, \gamma)$, each $G_{i+1}$ is a subgraph of $G_i$
formed by deleting one edge and removing any isolated vertices,
$\gamma_i$ is an odd assignment of charges to $G_i$,
and $\ell_i$ is a literal corresponding to an edge in $G_{i}\setminus G_{i+1}$.
Let
\[
\Gamma_i = \TS_{G_0, \gamma_0} \cup \{\ell_0\}
\cup \dots \cup
 \TS_{G_i, \gamma_i} \cup \{\ell_i\}.
\]
As we go we will construct an $\SPRnnv$ derivation containing sets of
clauses $\Gamma'_i$ extending and subsumed by $\Gamma_i$, and we will eventually
reach a stage $m$ where $\Gamma_m$ is trivially refutable.
The values of $\ell_i$, $G_{i+1}$ and $\gamma_{i+1}$ are defined from
$G_i$ and $\gamma_i$ according to the next three cases.

\emph{Case 1:} $G_i$ contains a vertex $j$ of degree 1.
Let $\{j,k\}$ be the edge touching~$j$.
If~$k$ has degree~2 or more,
we define $(G_{i+1}, \gamma_{i+1})$
by letting $G_{i+1}$ be $G_i$ with edge $\{j,k\}$ and vertex $j$
removed, and letting
$\gamma_{i+1}$ be $\gamma_i$ restricted to $G_{i+1}$ and with
$\gamma_{i+1}(k) = \gamma_i(k) + \gamma_i(j)$.
If $k$ has degree $1$
and the same charge as~$j$,
then we let $G_{i+1}$ be $G_i$ with both $j$ and $k$ removed
(with unchanged charges).
In both cases, every clause
in $\TS_{G_{i+1}, \gamma_{i+1}}$ is derivable from $\TS_{G_i, \gamma_i}$
by a $\vdash_1$ step, as the Tseitin condition on $j$
in $\TS_{G_i, \gamma_i}$ is a unit clause; we set $\ell_i$
to be the literal contained in this clause.
If $k$ has degree $1$ and opposite charge from $j$, then
we can already derive a contradiction from
$\TS_{G_i, \gamma_i}$ by one $\vdash_1$ step.

\emph{Case 2:} $G_i$ contains no vertices of degree 1 or 2.
Apply Lemma~\ref{lem:short_cycle} to find a cycle $K$ in $G_i$
of length at most $2 \log n$ and let $e$ be the first edge in $K$.
Our goal is to derive the unit clause $\olnot x_e$
and remove $e$ from $G_i$.

Let $\alpha$ be any assignment to the variables on $K$
which sets $x_e$ to 1, and let $\tau$ be the opposite assignment.
Using Lemma~\ref{lem:Tseitin_symmetry} applied simultaneously to all
graphs $G_0, \dots, G_i$ we have
${(\Gamma_i)}_{\rest\alpha} = {(\Gamma_i)}_{\rest\tau}$,
as the unit clauses $\ell_i$ are unaffected by these restrictions.
Hence by Lemma~\ref{lem:SPR_symmetry_2},
$\SPRnnv$ inferences can be used to
introduce all clauses $\olnot{\alpha}$, of which there are at most
$2^{2 \log n -1}$. We resolve them all together
to get the unit clause $\olnot x_e$.
This subsumes all other clauses introduced
so far in this step;
we set $\ell_i$ to be~$\olnot{x_e}$, and
by Lemma~\ref{lem:subsumeNoDelete}(a), we may 
ignore these subsumed clauses in future inferences. (Therefore we
avoid needing the deletion rule.)
We define $(G_{i+1}, \gamma_{i+1})$ by
deleting edge~$e$ from~$G_i$ and leaving $\gamma_i$ unchanged.
All clauses in $\TS_{G_{i+1}, \gamma_{i+1}}$ can now be derived
from $\TS_{G_i, \gamma_i}$ and $\olnot x_e$ by single $\vdash_1$ steps.

\emph{Case 3:} $G_i$ contains no vertices of degree~1, but
may contain vertices of degree~2. We will adapt the argument of case~2.
Redefine a \emph{path} to be a sequence of edges connected by
degree-2 vertices. By temporarily replacing paths in $G_i$ with edges,
we can apply Lemma~\ref{lem:short_cycle} to find a cycle $K$
in $G_i$ consisting of edge-disjoint paths $p_1, \dots, p_m$ where
$m \le 2 \log n$. Let $x_j$ be the variable associated with the first edge in $p_j$. For each~$j$,
there are precisely two assignments to the variables in $p_j$ which do not
immediately falsify some axiom of $\TS_{G_i, \gamma_i}$.
Let $\alpha$ be a partial assignment which picks one of these two
assignments for each $p_j$, and such that $\alpha(x_1)=1$.
As in case~2, $\SPRnnv$ inferences can be
used to introduce $\olnot{\alpha}$ for each $\alpha$ of this form.

Let us look at the part of $\olnot{\alpha}$ consisting of literals from path
$p_j$. This has the form $z^j_1 \lor \dots \lor z^j_r$, where
$z^j_1$ is $x_j$ with positive or negative sign and for each $k$, by the choice of $\alpha$, there are Tseitin axioms
expressing that $z^j_k$ and $z^j_{k+1}$ have the same value.
Hence if we set $z^j_1 =0$ we can set all literals in this clause
to $0$ by unit propagation.
Applying the same argument to all parts of $\alpha$
shows that
we can derive $z^1_1 \lor \dots \lor z^m_1$ from $\olnot \alpha$
and $\TS_{G_i, \gamma_i}$ with a single $\vdash_1$ step.
We introduce all $2^{m-1}$ such clauses, one for each $\alpha$,
all with $z^1_1 = \olnot{x_1}$.
We resolve them together to get the unit clause $\olnot{x_1}$,
then proceed as in case~2.

For the size bound, each case above requires us to derive at most
$n \! \cdot \! |\TS_{G,\gamma}|$ clauses, and the refutation can take at most $n$
steps.
\end{proof}

\subsection{Or-ification and xor-ification}\label{sec:PRorXor}

Orification and xorification have been widely used to make
hard instances of propositional tautologies,
see~\cite{BIW:nearoptimal,BenSasson:sizespace,Urquhart:regularresolution}.
This and the next section discuss how $\SPRnnv$ inferences can be used to ``undo'' the effects
of orification, xorification, and lifting without using
any new variables. As a consequence, these techniques are not likely to
be helpful in establishing lower bounds for the size of $\PRnnv$ refutations.

Typically, one ``orifies'' many variables at once; however, for the
purposes of this paper, we describe orification of a single variable.
Let $\Gamma$ be a set of clauses, and $x$~a variable.  For the $m$-fold
\emph{orification} of $x$, we introduce new variables $x_1,\ldots, x_m$,
with the intent of replacing $x$ with~$x_1\lor x_2\lor \cdots \lor x_m$.
Specifically, each clause $x\lor C$ in~$\Gamma$ is replaced with
$x_1\lor \cdots \lor x_m \lor C$, and each clause $\olnot x \lor C$
is replaced with the $m$-many clauses
$\olnot{x_j}\lor C$. Let $\Gamma^{\lor}$ denote the results of
this orification of~$x$. We claim that $\SPRnnv$ inferences may be used
to derive $\Gamma$ (with $x$ renamed to~$x_1$) from~$\Gamma^{\lor}$,
undoing the orification, as follows.
We first use $\SPRnnv$ inferences to derive each clause
$x_1 \lor \olnot{x_j}$ for $j>1$. This is done using Lemma~\ref{lem:SPR_symmetry_2},
with
$\alpha_j$ setting $x_1$ to $0$ and $x_j$ to $1$,
and $\tau_j$ setting $x_1$ to $1$ and $x_j$ to $0$,
so that $\tau_j$ is $\alpha_j$ with $x_1$ and $x_j$ swapped.
Thus any clause $x_1\lor \cdots \lor x_m \lor C$ in $\Gamma^{\lor}$
can be resolved with these to yield $x_1 \lor C$, and for  clauses
$\olnot{x_1} \lor C$ in $\Gamma^{\lor}$ we do not need to change anything.

Xorification of~$x$ is a similar construction, but now
we introduce $m$ new variables with the intent of letting
$x$ be expressed by $x_1 \oplus x_2 \oplus \cdots \oplus x_m$.
Each clause $x \lor C$ in~$\Gamma$
(respectively, $\olnot x \lor C$ in~$\Gamma$) is replaced by $2^{m-1}$ many
clauses $x_1^\sigma\lor x_2^\sigma\lor \cdots \lor x_m^\sigma \lor C$
where $\sigma$ is a partial assignment setting an odd number
(respectively, an even number) of the variables $x_j$ to~$\tTrue$.
To undo the xorification it is enough to derive the
unit clauses $\olnot {x_j}$ for $j>1$.
So for each $j>1$,  we first use  Lemma~\ref{lem:SPR_symmetry_2} to introduce
the clause~$x_1 \lor \olnot x_j$,
using the same partial assignments as in the previous paragraph,
and the clause~$\olnot{x_1} \lor \olnot{x_j}$,
using assignments
$\alpha_j$ setting $x_1$ and $x_j$ both to $1$,
and $\tau_j$ setting $x_1$ and $x_j$ both to $0$,
so that~$\tau_j$ is~$\alpha_j$ with the signs of both $x_1$ and $x_j$ flipped.
Resolving these gives~$\olnot x_j$.
This subsumes $x_1 \lor \olnot x_j$ and $\olnot{x_1} \lor \olnot{x_j}$,
so  by Lemma~\ref{lem:subsumeNoDelete}(a), 
we may ignore these two clauses in later $\SPRnnv$ steps, and can thus use
the same argument to derive the clauses $\olnot x_i$ for~\mbox{$i \neq j$}, since $\alpha_{i}$ and $\tau_{i}$ do
not affect the clause~$\olnot x_j$.

\subsection{Lifting}

Lifting is a technique for leveraging lower bounds on decision trees to obtain
lower bounds in stronger computational models,
see~\cite{RazMcKenzie:SeparationNC,BHP:HardnessAmplification,HuynhNordstrom:Amplifying}.

The most common form of lifting is the ``indexing gadget'' where
a single variable $x$ is replaced by $\ell + 2^\ell$
new variables $y_1,\ldots,y_\ell$
and $z_0,\ldots, z_{2^\ell-1}$. The intent is that
the variables $y_1,\ldots,y_\ell$ specify an integer~$i \in [2^\ell]$,
and $z_i$ gives the value of~$x$.
As in Section~\ref{sec:BPHP},
we write $(\vec y \pigeonto i)$ for the conjunction $\bigwedge_j (y_j = i_j)$
where $i_j$ is the $j$-th bit of $i$,
and write $(\vec y \notpigeonto i)$ for its negation $\bigvee_j (y_j \neq\penalty10000 i_j)$.
Thus, $x$~is equivalent to the CNF formula
$\bigwedge_{i \in [2^\ell]} \left( (\vec y \notpigeonto i) \lor z_i \right)$,
and $\olnot x$ is equivalent to the CNF
formula $\bigwedge_{i \in [2^\ell]} \left( (\vec y \notpigeonto i) \lor \olnot{z_i} \right)$.

Let $\Gamma$ is a set of clauses with an $\SPRnnv$ refutation.
The indexing gadget applied to~$\Gamma$
on the variable~$x$ does the following to modify $\Gamma$ to produce
set of lifted clauses~$\Gamma^\prime$:
Each clause $x \dotlor C$ containing~$x$ is replaced
by the $2^\ell$ clauses $(\vec y \notpigeonto i) \lor z_i \lor C$ for
$i\in[2^\ell]$, and each clause
$\olnot x \dotlor C$ containing~$\olnot x$ is replaced
by the $2^\ell$ clauses $(\vec y \notpigeonto i) \lor \olnot{z_i} \lor C$.

For all $i\not=0$ and all $a,b\in \{ 0, 1 \}$,  let $\alpha_{i,a,b}$ and $\tau_{i,a,b}$ be the partial
assignments
\begin{align*}
\alpha_{i,a,b} & := \, ( \vec y \pigeonto i ) \wedge \, z_0 = a \, \wedge \, z_i = b \\
\tau_{i,a,b} & := \, ( \vec y \pigeonto 0 ) \wedge \, z_0 = b \, \wedge \, z_i = a.
\end{align*}
Since $i \not= 0$ always holds, it is immediate
that conditions~(2) and~(3) of Lemma~\ref{lem:SPR_symmetry_2} hold.
For condition~(1), observe that the set of clauses
$\{ (\vec y \notpigeonto j) \lor z_j \lor C : j \in[2^\ell] \}$,
restricted by $( \vec y \pigeonto i )$, becomes the single clause
$z_i \lor C$, and restricted by $( \vec y \pigeonto 0)$
becomes $z_0 \lor C$.
In this way $\Gamma'_{\rest\alpha_{i,a,b}} = \Gamma'_{\rest\tau_{i,a,b}}$
and condition~1.\ also holds.
Therefore by Lemma~\ref{lem:SPR_symmetry_2}, $\SPRnnv$
inferences can be used to derive all clauses
$\olnot {\alpha_{i,a,b}}$, namely all the clauses
$(\vec y \notpigeonto i) \lor z_0 \! \neq \! a \lor z_i \! \neq \! b$.
For each fixed $i \neq 0$ this is four clauses, which can be resolved together
to give the clause $(\vec y \notpigeonto i)$.
Then from these $2^\ell-1$ clauses we can obtain by resolution
each unit clause~$y_j$ for $j=1, \ldots, \ell$.
Finally, using unit propagation with these, we derive the
clauses~$z_0 \dotlor C$ and $\olnot{z_0} \dotlor C$
for all original clauses $x \dotlor C$ and $\olnot x \dotlor C$ in~$\Gamma$.
We have thus derived from~$\Gamma^\prime$, using $\SPRnnv$ and resolution
inferences, a copy~$\Gamma^\pprime$ of all the clauses
in~$\Gamma$, except with~$x$ replaced with~$z_0$. The other clauses in
in~$\Gamma^\prime$ or that were inferred during the process of
deriving~$\Gamma^\pprime$ are subsumed
by either the unit clauses $y_j$ or the clauses in~$\Gamma^\pprime$.
Thus applying part~(b) and then part~(a)\ of Lemma~\ref{lem:subsumeNoDelete},
they do not interfere with
future $\SPRnnv$ inferences refuting~$\Gamma^\pprime$.


\section{Lower bounds}\label{sec:Lowerbounds}

This section gives an exponential separation between
$\DRATnnv$ and $\RATnnv$,
by showing that the bit pigeonhole principle $\BPHP_n$ requires
exponential size refutations in $\RATnnv$. This lower bound
still holds if we allow some deletions, as long as no initial
clause of $\BPHP_n$ is deleted.
On the other hand, with unrestricted deletions,
it follows from Theorems~\ref{thm:BCsimRATnoNew},~\ref{thm:DRATnnvDPRnnv} and~\ref{thm:SPR_BPHP}
in this paper that it has polynomial size refutations
in $\DRATnnv$ and even in $\DBCnnv$, as well as in $\SPRnnv$.

Kullmann~\cite{Kullmann:GeneralizationER} has already
proved related separations for generalized
extended resolution (GER), which lies somewhere between $\tDBC$
and $\tBC$ in strength.
That work shows separations between various subsystems of GER,
and in particular gives an exponential lower bound
on proofs of $\PHP_n$ in the system GER with no new
variables, by analyzing which clauses are blocked with
respect to $\PHP_n$.

We define the \emph{pigeon-width} of a clause or assignment
to equal the number of distinct pigeons that it mentions.
Our size lower bound for $\BPHP_n$ uses a conventional strategy:
we first show a width lower bound (on pigeon-width),
and then use a random restriction
to show that a proof of subexponential size can be made
into one of small pigeon-width. We do not aim for optimal constants.

We have to be careful about one technical point in the second step,
which is that $\RATnnv$ refutation size
does not in general behave well under restrictions,
as discussed in Section~\ref{sec:ER_nnv}.
So, rather than using restrictions as such to reduce width, we will define
a partial random matching~$\rho$ of pigeons to holes and show that if $\BPHP_n$ has a $\RATnnv$ refutation
of small size, then $\BPHP_n \cup \rho$ has one of small pigeon-width.

A useful tool in analyzing resolution derivations from a
set of clauses $\Gamma$ is the
\emph{Prover-Adversary game} on~$\Gamma$ (see e.g.~\cite{Pudlak:ProofsAsGames, atserias2008combinatorial}). 
In the game, the Adversary claims to know a
satisfying assignment for $\Gamma$,
and the Prover tries to force her into a contradiction by querying
the values of variables; the Prover can also forget variable assignments to
save memory and simplify his strategy.
A \emph{position} in the game is a partial assignment
$\alpha$ recording the contents of the Prover's memory.
To fully specify the game we also need to specify the
starting position.

As the next lemma shows, a strategy for the Prover in this game
is essentially the same thing as a resolution derivation.
But it is more intuitive to describe a strategy than a
derivation,
and the game also gives a natural way to show width lower bounds.

\begin{lem}\label{lem:Prover_Adversary}
Consider a restriction of the Prover-Adversary game on
a set of clauses~$\Gamma$ starting 
from position $\olnot C$ in which the Prover's
memory can hold information about at most~$m$ pigeons
simultaneously.
If the Prover has a winning strategy in this game,
then~$C$ is derivable from $\Gamma$ in pigeon-width $m$.
If the Adversary has a winning strategy, then
$C$ is not derivable from $\Gamma$ in pigeon-width $m-1$.
\end{lem}

\begin{proof}
We can think of a winning strategy for the Prover
as a tree in which the nodes are labelled
with a partial assignment (the position~$\alpha$) and with the
Prover's action in that situation, that is:
query a variable, forget
a variable, or declare victory because~$\alpha$ falsifies
a clause from~$\Gamma$. We can make this into a resolution refutation by
replacing each label~$\alpha$ with the clause negating it
and interpreting the three actions as respectively
resolution, weakening, and deriving a clause from an
axiom by weakening.

For the other direction, it is enough to
construct a winning strategy for the Prover from a
derivation of pigeon-width $m-1$. This is the reverse
of the process described above, except that we need
to be careful with the resolution rule. Suppose
an instance of the rule is: from $p \vee D$ and
 $\olnot p \vee E$ derive~$D \vee E$, where each
 clause mentions at most $m-1$ pigeons.
In the Prover-strategy, this becomes:
from position~$\olnot{D} \cup \olnot E$, query $p$. If it
is false, forget some variables from $\olnot E$ to reach
position~$\olnot{p} \cup \olnot{D}$. If it is true, forget
some variables from $\olnot D$
to reach position $p \cup \olnot E$. Thus the
Prover may be in a position mentioning $m$ pigeons
immediately after $p$ is queried.
\end{proof}

\begin{lem}\label{lem:adversary_width}
Let $\beta$ be a partial assignment corresponding to
a partial matching of $m$ pigeons to holes.
Then $\BPHP_n \cup \beta$ requires pigeon-width $n{-}m$ to refute in
resolution.
\end{lem}

\begin{proof}
$\BPHP_n \cup \beta$ is essentially an unusual
encoding of the pigeonhole principle with $n+1-m$ pigeons and $n-m$ holes.
Thus, if the Prover is limited to remembering variables
from at most $n-m$ pigeons, there is an easy strategy
for the Adversary in the game starting from the
empty position. Namely, she can always maintain a matching
between the pigeons mentioned in the Prover's memory
and the available holes.
The result follows by Lemma~\ref{lem:Prover_Adversary}.
\end{proof}


\begin{thm}\label{thm:DRAT_width}
Let $\rho$ be a partial matching of size at most $n/4$.
Let $\Pi$ be a $\DRATnnv$ refutation of $\BPHP_n \cup \rho$ in which
no clause of $\BPHP_n$ is ever deleted. Then
some clause in $\Pi$ has pigeon-width more than $n/3$.
\end{thm}

\begin{proof}
Suppose for a contradiction there is a such a refutation
$\Pi$ in pigeon-width $n/3$. We consider each $\tRAT$ inference
in $\Pi$ in turn, and show it
can be eliminated and replaced with standard resolution reasoning,
without increasing the pigeon-width.

Inductively suppose $\Gamma$ is a set of clauses derivable  from $\BPHP_n \cup \rho$
in pigeon-width~$n/3$, using only resolution and weakening.
Suppose a clause $C$ in $\Pi$ of the form $p \dotlor C'$ is $\tRAT$
with respect to
 $\Gamma$ and~$p$. 
Let $\alpha = \olnot C$, so
$\alpha(p)=0$ and
$\alpha$ mentions at most~$n/3$
pigeons. We consider three cases.

\emph{Case 1:} the assignment $\alpha$ is inconsistent with $\rho$.
This means that $\rho$ satisfies a literal which appears in $C$,
so $C$ can be derived from $\rho$ by a single weakening step.

\emph{Case 2:}
the assignment $\alpha \cup \rho$ can be extended to a
partial matching~$\beta$ of the pigeons it mentions.
We will show that this cannot happen.
Let $x$ be the pigeon
associated with the literal $p$. Let $y=\beta(x)$ and let $y'$ be the
hole~$\beta$ would map~$x$ to if the bit $p$ were flipped to $\tTrue$.
If $y' = \beta(x')$
for some pigeon~$x'$ in the domain of~$\beta$, let $\beta'=\beta$. Otherwise let
$\beta' = \beta \cup \{ (x',y')\}$ for some pigeon~$x'$ outside the domain of~$\beta$.

Let $H$ be the hole axiom
$(x \notpigeonto y') \lor (x' \notpigeonto y')$ in $\Gamma$.
The clause $(x \notpigeonto y')$
contains the literal~$\olnot p$,
since $(x \pigeonto y')$ contains $p$.
So  $H = \olnot{p} \dotlor H'$
for some clause $H'$.
By the $\tRAT$ condition, either $C' \cup H'$ is a tautology or $\Gamma \vdash_1 C \lor H'$.
Either way, $\Gamma \cup \olnot{C} \cup \olnot{H'} \vdash_1 \bot$.
Since $\beta' \supseteq \alpha$, $\beta'$ falsifies~$C$.
It also falsifies~$H'$,
since it satisfies $(x \pigeonto y') \wedge (x' \pigeonto y')$
except at~$p$.
It follows that $\Gamma \cup \beta' \vdash_1 \bot$.
By assumption,
$\Gamma$ is derivable from $\BPHP_n \cup \rho$
in pigeon-width~$n/3$, and $\beta'$ extends $\rho$.
Since unit propagation does not increase pigeon-width,
this implies that $\BPHP_n \cup \beta'$
is refutable in resolution in pigeon-width $n/3$,
by first deriving $\Gamma$ and then using unit propagation.
This contradicts
Lemma~\ref{lem:adversary_width} as $\beta'$ is a matching of
at most $n/3+n/4+1$ pigeons.

\emph{Case 3:} the assignment $\alpha \cup \rho$ cannot
be extended to a partial matching of the pigeons it mentions.
Consider the Prover-Adversary game on $\BPHP_n \cup \rho$
with starting position~$\alpha$. The Prover can ask all remaining
bits of the pigeons mentioned in~$\alpha$, and
since there is no suitable partial matching this
forces the Adversary to reveal a collision and lose the game.
This strategy
has pigeon-width $n/3$; it follows that $C$ is derivable from
$\BPHP_n \cup \rho$ in resolution in this pigeon-width, as required.
\end{proof}

\begin{thm}\label{thm:BPHP_size}
Let $\Pi$ be a $\DRATnnv$ refutation of $\BPHP_n$ in which
no clause of $\BPHP_n$ is ever deleted.
Then $\Pi$ has size at least~$2^{n/60}$.
\end{thm}

\begin{proof}
Construct a random restriction $\rho$ by selecting each
pigeon independently with probability $1/5$ and then randomly
matching the selected pigeons with distinct holes
(there is an ${(1/5)}^{n+1}$ chance that there is no matching,
because we selected all the pigeons --- in this case we set all
variables at random).

Let $m=n/4$. Let $C$ be a clause mentioning
at least $m$ distinct pigeons $x_1, \dots, x_m$
and choose literals $p_1, \dots, p_m$ in $C$ such that $p_i$
belongs to pigeon $x_i$.
The probability that~$p_i$ is satisfied by $\rho$ is~$1/10$.
However, these events are not quite independent for different~$i$,
as the holes used by other pigeons are blocked for pigeon~$x_i$.
To deal with this, we may assume that pigeons $x_1, \dots, x_m$,
in that order,  were
the first pigeons considered in the construction of $\rho$.
When we come to $x_i$, if we set it, then
there are $n/2$ holes which would satisfy $p_i$, at
least $n/2 - m \ge n/4$ of which are free;
so of the free holes, the fraction which satisfy $p_i$
is  at least $1/3$.
So the probability that $\rho$ satisfies $p_i$,
conditioned on it not satisfying any of $p_1, \dots, p_{i-1}$,
is at least $1/15$.
%
%
Therefore the probability that~$C$
is not satisfied by $\rho$ is at most ${(1-1/15)}^m < e^{-m/15} = e^{-n/60}$.

Now suppose $\Pi$ contains no more than $2^{n/60}$ clauses
of pigeon width at least~$n/4$. 
By the union bound, for a random $\rho$,
the probability that at least one of these clauses
is not satisfied by $\rho$ is at most $2^{n/60} e^{-n/60} = {(2/e)}^{n/60}$.
Therefore most restrictions~$\rho$ satisfy
all clauses in $\Pi$ of pigeon-width at least~$n/4$, and by the Chernoff
bound we may choose such a  $\rho$
which also sets no more than $n/4$
pigeons.

We now observe inductively that for each clause $C$ in $\Pi$,
some subclause of~$C$
is derivable from $\BPHP_n \cup \rho$ in resolution in
pigeon-width $n/3$,
ultimately contradicting Lemma~\ref{lem:adversary_width}.
If $C$ has pigeon-width more than $n/3$, this follows because $C$ is subsumed by~$\rho$.
Otherwise, if~$C$ is derived by a $\tRAT$ inference, we
repeat the proof of Theorem~\ref{thm:DRAT_width};
in case~2 we additionally use the observation that if $\Gamma \vdash_1 C \lor H'$
and $\Gamma'$ subsumes $\Gamma$, then
$\Gamma' \vdash_1 C \lor H'$.
\end{proof}

\begin{cor}\label{coro:RATminusnotSimulateTwo}
$\RATnnv$ does not simulate $\DRATnnv$.
$\RATnnv$ does not simulate $\SPRnnv$.
\end{cor}

\begin{proof}
By Theorem~\ref{thm:SPR_BPHP}, $\BPHP_n$ has short proofs
in $\SPRnnv$. Thus, by Theorem~\ref{thm:DRATnnvDPRnnv},
this also holds for $\DRATnnv$ (and for $\DBCnnv$ by
Theorem~\ref{thm:BCsimRATnoNew}). On the other hand,
Theorem~\ref{thm:BPHP_size} just showed $\BPHP_n$ requires
exponential size $\RATnnv$ proofs.
\end{proof}

\section{Open problems}\label{sec:open}

There are a number of open questions about the systems with no new variables.
Of particular importance is the question of the relative strengths of $\DPRnnv$,
$\DSRnnv$ and related systems. The results of~\cite{HeuleBiere:Variable,HKB:NoNewVariables,HKB:StrongExtensionFree} and
the present paper show that~$\DPRnnv$, and even the possibly weaker system $\SPRnnv$, are strong.
$\DPRnnv$ is a promising system
for effective proof search algorithms, but it is open whether practical proof search
algorithms can effectively exploit its strength.
It is also open whether $\DPRnnv$ or $\DSRnnv$
simulates $\tER$.

Another important question is
to understand the strength of deletion for these systems.
Of course, deletion is well-known to help the performance of SAT solvers in
practice, if for no other reason, because unit propagation is
faster when fewer clauses are present.
In addition, for systems such as~$\tRAT$, it is known that deletion
can allow new inferences. The results in
Sections~\ref{sec:Upperbounds} and~\ref{sec:Lowerbounds}
improve upon this
by showing that $\RATnnv$ does not simulate $\DRATnnv$. This
strengthens the case, at least in theory, for the importance of deletion.

In Section~\ref{sec:Upperbounds} we
described small $\SPRnnv$ proofs of  many of the
known ``hard'' tautologies that have been shown to require
exponential size proofs in constant depth Frege. It is open
whether $\SPRnnv$ simulates Frege; and by these results,
any separation of $\SPRnnv$ and Frege systems will likely require
developing new techniques.  Even more tantalizing, we can ask
whether $\SRnnv$ simulates Frege.

There are several hard tautologies for which we do not
whether there are as polynomial size $\SPRnnv$ proofs.
Jakob Nordstr\"om [personal communication, 2019] suggested
(random) 3-XOR SAT
and the even coloring principle as examples.
3-XOR SAT has short cutting planes proofs via Gaussian elimination;
it is open whether $\SPRnnv$ or $\DSPRnnv$ or even $\DSRnnv$ has polynomial
size refutations for all unsatisfiable 3-XOR SAT principles.
The even coloring principle
is a special case of the Tseitin principle~\cite{Markstrom:EvenColoring}:
the graph has an odd number
of edges, each vertex
has even degree, and the initial clauses assert that, for
each vertex, exactly one-half the incident
edges are labeled~1.  It is not hard to see that the even coloring principle can
be weakened to the Tseitin principle by removing some clauses
with the deletion rule.
Hence there are polynomial size $\DSPRnnv$ refutations (with deletion)
of the even coloring principle. It is open whether $\SPRnnv$ (without
deletion) has polynomial size refutations for the even coloring
principle.

Paul Beame [personal communication, 2018] suggested that the
graph PHP principles (see~\cite{BenSassonWigderson:ShortProofsNarrow})
may separate systems such as $\SPRnnv$ or even $\SRnnv$ from
Frege systems. However, there are reasons to suspect that
in fact the graph PHP principles also have short $\SPRnnv$
proofs. Namely, $\tSPR$ inferences
can infer a lot of clauses from the graph PHP clauses. If an
instance of graph PHP has every pigeon with outdegree $\ge 2$,
then there must be an alternating cycle of pigeons $i_1,\ldots i_{\ell+1}$
and holes $j_1,\ldots j_\ell$ such that $i_\ell = i_1$,
the edges $(i_s, j_s)$ and $(i_{s+1},j_s)$ are
all in the graph, and $\ell = O(\log n)$. Then
an $\tSPR$ inference can be used to learn the clause
$\olnot{x_{i_1,j_1}} \lor \olnot{x_{i_2,j_2}} \lor \cdots \lor \olnot{x_{i_\ell,j_\ell}}$,
by using the fact that a satisfying assignment that falsifies this clause
can be replaced by the assignment that maps instead each
pigeon~$i_{s+1}$ to hole~$j_s$.

This construction clearly means that $\tSPR$ inferences can infer
many clauses from the graph PHP clauses. However, we do not know how
to use these to form a short $\SPRnnv$ proof of the graph PHP principles.
It remains open whether a polynomial size $\SPRnnv$ proof exists.

\paragraph{Acknowledgements}
We thank the reviewers of the conference and journal versions
of this paper for suggestions and comments that improved the paper.
We also thank Jakob Nordstr\"om, Paul Beame, Marijn Heule,
Thomas Kochmann and Oliver Kullmann for useful comments, questions and suggestions.

\bibliographystyle{alpha}
\bibliography{logic}

\newcommand{\etalchar}[1]{$^{#1}$}
\begin{thebibliography}{HHJW13b}

\bibitem[AD08]{atserias2008combinatorial}
A.~Atserias and V.~Dalmau.
\newblock A combinatorial characterization of resolution width.
\newblock {\em Journal of Computer and System Sciences}, 74(3):323--334, 2008.

\bibitem[Ajt90]{Ajtai:ParityandPHP}
M.~Ajtai.
\newblock Parity and the pigeonhole principle.
\newblock In {\em Feasible Mathematics: A {M}athematical {S}ciences {I}nstitute
  Workshop held in {I}thaca, {N}ew {Y}ork, {J}une 1989}, pages 1--24.
  Birkh{\"a}user, 1990.

\bibitem[BHP10]{BHP:HardnessAmplification}
Paul Beame, Trinh Huynh, and Toniann Pitassi.
\newblock Hardness amplification in proof complexity.
\newblock In {\em Proc. 42nd Annual ACM Symposium on Theory of Computing
  (STOC'10)}, pages 87--96, 2010.

\bibitem[BIK{\etalchar{+}}96]{BIKPP:nullstellensatz}
Paul Beame, Russell Impagliazzo, Jan Kraj{\'\i\v c}ek, Toniann Pitassi, and
  Pavel Pudl\'ak.
\newblock Lower bounds on {H}ilbert's {N}ullstellensatz and propositional
  proofs.
\newblock {\em Proceedings of the London Mathematical Society}, 73(3):1--26,
  1996.

\bibitem[BKS04]{BKS:clauselearning}
Paul Beame, Henry~A. Kautz, and Ashish Sabharwal.
\newblock Towards understanding and harnessing the potential of clause
  learning.
\newblock {\em Journal of Artificial Intelligence Research}, 22:319--351, 2004.

\bibitem[BS09]{BenSasson:sizespace}
Eli Ben-Sasson.
\newblock Size space tradeoffs for resolution.
\newblock {\em SIAM Journal on Computing}, 38(6):2511--2525, 2009.

\bibitem[BSIW04]{BIW:nearoptimal}
Eli Ben-Sasson, Russell Impagliazzo, and Avi Wigderson.
\newblock Near optimal separation of tree-like and general resolution.
\newblock {\em Combinatorica}, 24(4):585--603, 2004.

\bibitem[BSW01]{BenSassonWigderson:ShortProofsNarrow}
Eli Ben-Sasson and Ave Wigderson.
\newblock Short proofs are narrow --- resolution made simple.
\newblock {\em Journal of the ACM}, 48:149--169, 2001.

\bibitem[BT19]{BussThapen:DratAndPr_SAT}
Sam Buss and Neil Thapen.
\newblock {DRAT} proofs, propagation redundancy, and extended resolution.
\newblock In {\em Proc. 22nf Intl. Conference on Theory and Applications of
  Satisfiability Testing (SAT 2019)}, Springer-Verlag Lecture Notes in Computer
  Science 11628, pages 71--89, 2019.

\bibitem[Bus86]{Buss:bookBA}
Samuel~R. Buss.
\newblock {\em Bounded Arithmetic}.
\newblock Bibliopolis, Naples, Italy, 1986.
\newblock Revision of 1985 Princeton University Ph.D. thesis.

\bibitem[Cha70]{Chang:UnitAndInput}
C.~L. Chang.
\newblock The unit proof and the input proof in theorem proving.
\newblock {\em J. ACM}, 17(4):698--707, 1970.

\bibitem[Coo75]{Cook:PV}
Stephen~A. Cook.
\newblock Feasibly constructive proofs and the propositional calculus.
\newblock In {\em Proceedings of the Seventh Annual ACM Symposium on Theory of
  Computing}, pages 83--97. Association for Computing Machinery, 1975.

\bibitem[CR74]{CookReckhow:proofsstoc}
Stephen~A. Cook and Robert~A. Reckhow.
\newblock On the lengths of proofs in the propositional calculus, preliminary
  version.
\newblock In {\em Proceedings of the Sixth Annual ACM Symposium on the Theory
  of Computing}, pages 135--148, 1974.

\bibitem[CR79]{CookReckhow:proofs}
Stephen~A. Cook and Robert~A. Reckhow.
\newblock The relative efficiency of propositional proof systems.
\newblock {\em Journal of Symbolic Logic}, 44:36--50, 1979.

\bibitem[GN03]{GoldbergNovikov:verification}
Evguenii~I. Goldberg and Yakov Novikov.
\newblock Verification of proofs of unsatisfiability for {CNF} formulas.
\newblock In {\em Design, Automation and Test in Europe Conference (DATE)},
  pages 10886--10891. IEEE Computer Society, 2003.

\bibitem[GS88]{GrollmannSelman:CryptoMeasures}
Joachim Grollmann and Alan Selman.
\newblock Complexity measures for public-key cryptosystems.
\newblock {\em SIAM Journal on Computing}, 17(2):309--335, 1988.

\bibitem[HB18]{HeuleBiere:Variable}
Marijn J.~H. Heule and Armin Biere.
\newblock What a difference a variable makes.
\newblock In {\em Tools and Algorithms for the Construction and Analysis of
  Systems - 24th International Conference (TACAS 2018)}, Lecture Notes in
  Computer Science 10806, pages 75--92. Springer Verlag, 2018.

\bibitem[HHJW13a]{HHW:trimming}
Marijn J.~H. Heule, Warren~A. Hunt~Jr., and Nathan Wetzler.
\newblock Trimming while checking clausal proofs.
\newblock In {\em Formal Methods in Computer-Aided Design (FMCAD)}, pages
  181--188. IEEE, 2013.

\bibitem[HHJW13b]{HHW:verifying}
Marijn J.~H. Heule, Warren~A. Hunt~Jr., and Nathan Wetzler.
\newblock Verifying refutations with extended resolution.
\newblock In {\em Automated Deduction - 24th International Conference (CADE)},
  Lecture Notes in Computer Science 7898, pages 345--359. Springer Verlag,
  2013.

\bibitem[HKB17]{HKB:NoNewVariables}
Marijn J.~H. Heule, Benjamin Kiesl, and Armin Biere.
\newblock Short proofs without new variables.
\newblock In {\em Automated Deduction - 26th International Conference (CADE)},
  Lecture Notes in Computer Science 10395, pages 130--147. Springer Verlag,
  2017.

\bibitem[HKB19]{HKB:StrongExtensionFree}
Marijn J.~H. Heule, Benjamin Kiesl, and Armin Biere.
\newblock Strong extension-free proof systems.
\newblock {\em Journal of Automated Reasoning}, pages 1--22, 2019.
\newblock Extended version of \cite{HKB:NoNewVariables}.

\bibitem[HKSB17]{HKSB:PRuning}
Marijn J.~H. Heule, Benjamin Kiesl, Martina Seidl, and Armin Biere.
\newblock {PR}uning through satisfaction.
\newblock In {\em Hardware and Software: Verification and Testing - 13th
  International Haifa Verification Conference (HVC)}, Lecture Notes in Computer
  Science 10629, pages 179--194. Springer Verlag, 2017.

\bibitem[HN12]{HuynhNordstrom:Amplifying}
Trinh Huynh and Jakob Nordstr{\"o}m.
\newblock On the virtue of succinct proofs: Amplifying communication complexity
  hardness to time-space trade-offs in proof complexity ({E}xtended abstract).
\newblock In {\em Proc. 44th Annual ACM Symposium on Theory of Computing
  (STOC'12)}, pages 233--248, 2012.

\bibitem[JHB12]{JHB:inprocessing}
Matti J{\"a}rvisalo, Marijn J.~H. Heule, and Armin Biere.
\newblock Inprocessing rules.
\newblock In {\em Automated Reasoning - 6th International Joint Conference
  (IJCAR)}, Lecture Notes in Computer Science 7364, pages 355--270. Springer
  Verlag, 2012.

\bibitem[KPW95]{KPW:PHP}
Jan Kraj\'\i\v{c}ek, Pavel Pudl\'ak, and Alan Woods.
\newblock Exponential lower bound to the size of bounded depth {F}rege proofs
  of the pigeonhole principle.
\newblock {\em Random Structures and Algorithms}, 7:15--39, 1995.

\bibitem[Kra97]{Krajicek:interpolation}
Jan Kraj\'\i\v{c}ek.
\newblock Interpolation theorems, lower bounds for proof systems, and
  independence results for bounded arithmetic.
\newblock {\em Journal of Symbolic Logic}, 62:457--486, 1997.

\bibitem[KRPH18]{KRPH:erDRAT}
Benjamin Kiesl, Adri{\'a}n Rebola-Pardo, and Marijn J.~H. Heule.
\newblock Extended resolution simulates {DRAT}.
\newblock In {\em Automated Reasoning - 6th International Joint Conference
  (IJCAR)}, Lecture Notes in Computer Science 10900, pages 516--531. Springer
  Verlag, 2018.

\bibitem[Kul97]{Kullmann:WorstCase3SAT}
Oliver Kullmann.
\newblock Worst-case analysis, 3-{SAT} decision and lower bounds: {A}pproaches
  for improved {SAT} algorithms.
\newblock {\em DIMACS Series in Discrete Mathematics and Theoretical Computer
  Science}, 35:261--313, 1997.

\bibitem[Kul99a]{Kullmann:NewMethods3SAT}
Oliver Kullmann.
\newblock New methods for 3-{SAT} decision and worst-case analysis.
\newblock {\em Theoretical Computer Science}, 223(1):1--72, 1999.

\bibitem[Kul99b]{Kullmann:GeneralizationER}
Oliver Kullmann.
\newblock On a generalizaton of extended resolution.
\newblock {\em Discrete Applied Mathematics}, 96-97:149--176, 1999.

\bibitem[Mar06]{Markstrom:EvenColoring}
Klas Markstr{\"o}m.
\newblock Locality and hard {SAT}-instances.
\newblock {\em Journal on Satisfiability, Boolean Modeling and Computation
  (JSAT)}, 2:221--227, 2006.

\bibitem[PBI93]{PBI:PHP}
Toniann Pitassi, Paul Beame, and Russell Impagliazzo.
\newblock Exponential lower bounds for the pigeonhole principle.
\newblock {\em Computational Complexity}, 3:97--140, 1993.

\bibitem[PS10]{PitassiSanthanam:effectivePsim}
Toniann Pitassi and Rahul Santhanam.
\newblock Effectively polynomial simulations.
\newblock In {\em Innovations in Computer Science (ICS)}, pages 370--381, 2010.

\bibitem[Pud97]{Pudlak:monotone}
Pavel Pudl\'ak.
\newblock Lower bounds for resolution and cutting planes proofs and monotone
  computations.
\newblock {\em Journal of Symbolic Logic}, 62(3):981--998, 1997.

\bibitem[Pud00]{Pudlak:ProofsAsGames}
Pavel Pudl\'ak.
\newblock Proofs as games.
\newblock {\em American Mathematical Monthly}, 107(6):541--550, 2000.

\bibitem[Pud03]{Pudlak:NPpairs}
Pavel Pudl\'ak.
\newblock On reducibility and symmetry of disjoint {NP} pairs.
\newblock {\em Theoretical Computer Science}, 205:323--339, 2003.

\bibitem[Raz94]{Razborov:NPpairs}
Alexander~A. Razborov.
\newblock On provably disjoint {NP}-pairs.
\newblock Technical Report TR94-006, Electronic Colloquium in Computational
  Complexity (ECCC), 1994.
\newblock Also available as BRIC technical report RS-94-36.

\bibitem[RM99]{RazMcKenzie:SeparationNC}
Ran Raz and Pierre McKenzie.
\newblock Separation of the monotone {NC} hierarchy.
\newblock {\em Combinatorica}, 19(3):403--435, 1999.

\bibitem[RS18]{RebolaPardoSuda:SatPreserving}
Adri{\'{a}}n {Rebola-Pardo} and Martin Suda.
\newblock A theory of satisfiability-preserving proofs in {SAT} solving.
\newblock In {\em Proc., 22nd Intl. Conf. on Logic for Programming, Artificial
  Intelligence and Reasoning (LPAR'22)}, EPiC Series in Computing 57, pages
  583--603. EasyChair, 2018.

\bibitem[SW83]{SiekmannWrightson:automation}
Jorg Siekmann and Graham Wrightson.
\newblock {\em Automation of Reasoning}, volume 1\&2.
\newblock Springer-Verlag, Berlin, 1983.

\bibitem[Sze03]{Szeider:Homomorphisms}
Stefan Szeider.
\newblock Homomorphisms of conjunctive normal forms.
\newblock {\em Discrete Applied Mathematics}, 130:351--365, 2003.

\bibitem[Tse68]{Tseitin:derivation}
G.~S. Tsejtin.
\newblock On the complexity of derivation in propositional logic.
\newblock {\em Studies in Constructive Mathematics and Mathematical Logic},
  2:115--125, 1968.
\newblock Reprinted in: {\protect{\cite[vol~2]{SiekmannWrightson:automation}}},
  pp.~466-483.

\bibitem[Urq87]{Urquhart:hardresolution}
Alasdair Urquhart.
\newblock Hard examples for resolution.
\newblock {\em Journal of the ACM}, 34:209--219, 1987.

\bibitem[Urq11]{Urquhart:regularresolution}
Alasdair Urquhart.
\newblock A near-optimal separation of regular and general resolution.
\newblock {\em SIAM Journal on Computing}, 40(1):107--121, 2011.

\bibitem[{Van}08]{VanGelder:RUP}
Allen {Van Gelder}.
\newblock Verifying {RUP} proofs of propositional unsatisfiability.
\newblock In {\em 10th International Symposium on Artificial Intelligence and
  Mathematics (ISAIM)}, 2008.
\newblock http://isaim2008.unl.edu/index.php?page=proceedings.

\bibitem[WHHJ14]{WHH:DRATtrim}
Nathan Wetzler, Marijn J.~H. Heule, and Warren~A. Hunt~Jr.
\newblock {DRAT}-trim: Efficient checking and trimming using expressive clausal
  proofs.
\newblock In {\em Theory and Applications of Satisfiability Testing - 17th
  International Conference (SAT)}, Lecture Notes in Computer Science 8561,
  pages 422--429. Springer Verlag, 2014.

\end{thebibliography}
\end{document}